\documentclass[conference]{IEEEtran}

\usepackage{amsmath}
\usepackage{amssymb}
\usepackage{amsthm}
\usepackage{paralist}
\usepackage{listings}
\usepackage{hyperref}
\usepackage{color}
\usepackage{graphicx}
\usepackage{float}
\usepackage{caption}
\usepackage{textcomp}
\usepackage{xspace}
\usepackage{pifont}
\usepackage{subcaption}
\usepackage{multirow}
\usepackage{algorithm,algpseudocode}
\usepackage{pgfplots,booktabs}
\usepackage{enumitem}
\usepackage{authblk}
\usepackage{balance}
\usetikzlibrary{snakes}
\tikzstyle{printersafe}=[snake=snake,segment amplitude=0 pt]

\begin{document}

\newcommand{\attacker}{\textit{A}\xspace}
\newcommand{\victim}{\textit{V}\xspace}
\newcommand{\provider}{\textit{T}\xspace}

\newcommand{\cmark}{\ding{51}}%
\newcommand{\xmark}{\ding{55}}%

\lstdefinelanguage{JavaScript}{
  keywords={typeof, new, true, false, catch, function, return, null, catch, switch, var, if, in, while, do, else, case, break},
  keywordstyle=\color{blue}\bfseries,
  ndkeywords={class, export, boolean, throw, implements, import, this},
  ndkeywordstyle=\color{darkgray}\bfseries,
  identifierstyle=\color{black},
  sensitive=false,
  comment=[l]{//},
  morecomment=[s]{/*}{*/},
  commentstyle=\color{purple}\ttfamily,
  stringstyle=\color{red}\ttfamily,
  morestring=[b]',
  morestring=[b]"
}

\lstset{
   language=JavaScript,
   extendedchars=true,
   basicstyle=\footnotesize\ttfamily,
   showstringspaces=false,
   showspaces=false,
   numbers=left,
   numberstyle=\footnotesize,
   numbersep=9pt,
   tabsize=2,
   breaklines=true,
   showtabs=false,
   captionpos=b,
   upquote=true
}

\definecolor{codegray}{gray}{0.95}
\newcommand{\code}[1]{\colorbox{codegray}{\texttt{#1}}}

\theoremstyle{definition}
\newtheorem{definition}{Definition}

\newcommand{\NN}{\mathbb{N}}
\newcommand{\sizeof}[1]{\left| #1 \right|}
\newcommand{\logz}[1]{\log_2{#1}}
\newcommand{\cachesize}{\mathit{size}}
\newcommand{\assoc}{\mathit{a}}
\newcommand{\genassoc}{p}
\newcommand{\archbits}{n}
\newcommand{\pagebits}{p}
\newcommand{\hugepagebits}{h}
\newcommand{\attackerbits}{\gamma}
\newcommand{\linebits}{\ell}
\newcommand{\tagbits}{t}
\newcommand{\indexbits}{c}
\newcommand{\slicebits}{s}
\newcommand{\setclass}[1]{\left[#1\right]}

\newcommand{\projset}[1]{\mathit{set}(#1)}
\newcommand{\projslice}[1]{\mathit{slice}(#1)}
\newcommand{\vptrans}{\mathit{pt}}


\newcommand{\Exp}[1]{E(#1)}
\newcommand{\Var}[1]{\mathit{Var}(#1)}


\newcommand{\oo}{\mathcal{O}}
\newcommand{\testset}{T}
\newcommand{\posset}{P}
\newcommand{\univset}{S}
\newcommand{\lowth}{l}
\newcommand{\upth}{u}


\newtheorem{theorem}{Theorem}
\newtheorem{lemma}[theorem]{Lemma}
\newtheorem{proposition}{Proposition}
\newtheorem{corollary}[theorem]{Corollary}
\theoremstyle{definition}
\newtheorem{example}{Example}

\newcommand{\skylake}{sky}
\newcommand{\haswell}{has}

\title{Theory and Practice of Finding Eviction Sets}


\author[1,3]{Pepe Vila}
\author[2]{Boris K{\"o}pf}
\author[1]{Jos\'e F. Morales}
\affil[1]{IMDEA Software Institute}
\affil[2]{Microsoft Research}
\affil[3]{Technical University of Madrid (UPM)}
\renewcommand\Authands{ and }

\maketitle

\begin{abstract}
  Many micro-architectural attacks rely on the capability of an
  attacker to efficiently find small {\em eviction sets}: groups of
  virtual addresses that map to the same cache set. This capability
  has become a decisive primitive for cache side-channel, rowhammer,
  and speculative execution attacks. Despite their importance,
  algorithms for finding small eviction sets have not been
  systematically studied in the literature.

  In this paper, we perform such a systematic study. We begin by
  formalizing the problem and analyzing the probability that a set of
  random virtual addresses is an eviction set. We then present novel
  algorithms, based on ideas from threshold group testing, that reduce
  random eviction sets to their minimal core in linear time, improving
  over the quadratic state-of-the-art.

  We complement the theoretical analysis of our algorithms with a
  rigorous empirical evaluation in which we identify and isolate
  factors that affect their reliability in practice, such as adaptive cache
  replacement strategies and TLB thrashing.  Our results indicate that
  our algorithms enable finding small eviction sets much faster than
  before, and under conditions where this was previously deemed
  impractical.
\end{abstract}

\section{Introduction}\label{sec:introduction}

Attacks against the micro-architecture of modern CPUs have rapidly
evolved from an academic stunt to a powerful tool in the hand of
real-world adversaries. Prominent examples of attacks include
side-channel attacks against shared CPU
caches~\cite{YaromFlushReload2014}, fault injection attacks against
DRAM~\cite{SeabornRH2015}, and covert channel attacks that leak
information from speculative executions~\cite{spectre18}.

A key requirement for many of the documented attacks is that the
adversary be able to bring specific cache sets into a controlled state.
For example, \emph{flush+reload}~\cite{YaromFlushReload2014} attacks
use special instructions to invalidate targeted cache content (like
\verb!clflush!  on x86), for which they require privileged execution
and shared memory space.  Another class of attacks, called
\emph{prime+probe}, evicts cache content by replacing it with new
content and can be performed without privileges from user space or from
a sandbox.

The primitive used for replacing cache content is called an {\em
  eviction set}. Technically, an eviction set is a collection of
(virtual) addresses that contains at least as many elements that map
to a specific cache set as the cache has {\em ways}. The intuition is that,
when accessed, an eviction set clears all previous content from the
cache set. Eviction sets enable an adversary to (1) bring specific
cache sets into a controlled state; and to (2) probe whether this
state has been modified by the victim, by measuring latency of
accesses to the eviction set.

Accessing a large enough set of virtual addresses is sufficient for
evicting any content from the cache. However, such large eviction sets
increase the time required for evicting and probing, and they
introduce noise due to the unnecessary memory accesses.
For targeted and stealthy eviction of cache content one hence seeks to
identify eviction sets of {\em minimal size}, which is fundamental,
for example, for
\begin{itemize}
\item fine-grained monitoring of memory usage by a concurrent process
  in timing attacks against last-level caches~\cite{LiuLLC2015,
    Irazoqui2015};
\item enforcing that memory accesses hit DRAM instead of the cache with
high enough frequency to flip bits in rowhammer attacks~\cite{GrussRow2016};and
\item increasing the number of instructions that are speculatively
  executed by ensuring that branch guards are not
  cached~\cite{spectre18}.
\end{itemize}

Computing minimal eviction sets is recognized as a challenging
problem, equivalent to learning the mapping from virtual addresses to
cache sets~\cite{LiuLLC2015}.  The difficulty of the problem is
governed by the amount of control the adversary has over the bits of
physical addresses.  For example, on bare metal, the adversary
completely controls the mapping to cache sets; on huge pages, it
controls the mapping to cache sets within each cache slice, but not
the mapping to slices; on regular $4$KB pages, it only partially
controls the mapping to sets within each slice; and on sandboxed or
hardened environments it may not have any control over the mapping at
all~\cite{SchwarzJSZero2018,Irazoqui2015}.

Several approaches in the literature discuss algorithms for finding
minimal eviction sets, see Section~\ref{sec:related} for an overview.
These algorithms rely on a two-step approach in which one first
collects a large enough set of addresses that is an eviction set, and
then successively reduces this set to its minimal core.
Unfortunately, these algorithms are usually only considered as a means
to another end, such as devising a novel attack. As a
consequence, they still lack an in-depth analysis in terms of
complexity, real-time performance, correctness, and scope, which
hinders progress in research on attacks and on principled
countermeasures at the same time.

\paragraph*{Our approach}
In this paper we perform the first systematic study of finding minimal
eviction sets as an algorithmic problem. In our study we proceed as
follows:
\begin{asparaitem}
\item We give the first {\em formalization and analysis} of the
  pro\-blem of finding eviction sets. We study different variants of
  the problem, corresponding to different goals, for example, ``evicting a
  specific cache set'', and ``evicting an arbitrary cache set''.
  For these goals,
  we express the probability that a set of virtual addresses is
  an eviction set as a function of its size. The function exhibits that a
  small set of virtual addresses is unlikely to be an eviction set, but that
  the likelihood grows fast with the set size. This analysis justifies
  the two-step approach taken in the literature for computing minimal
  eviction sets, and it exhibits favorable set sizes to start the
  reduction.

\item We design {\em novel algorithms} for finding minimal eviction
  sets.  The basis of our algorithms are tests from the
  literature~\cite{LiuLLC2015} that use the cache side-channel as an
  oracle to determine whether a given set of virtual addresses is an
  eviction set. The key observation underlying our algorithms is that
  these tests can be seen as so-called {\em threshold group
    tests}~\cite{Damaschke2006}. This observation allows us to
  leverage ideas from the group testing literature for computing
  minimal eviction sets. We show that the resulting algorithm reduces
  an eviction set of size $n$ to its minimal core using only $\oo(n)$
  memory accesses, which improves over the current $\oo(n^2)$
  state-of-the-art~\cite{OrenSpy2015}.

\item We perform a rigorous {\em reliability analysis} of our
  algorithms on Intel's Haswell and Skylake microarchitectures. In our
  analysis, we identify ways to isolate the influence of TLBs and
  cache replacement policies. This allows us to exhibit conditions
  under which our algorithms are almost perfectly reliable, as well as
  conditions under which their reliability degrades.

\item We carry out a {\em performance analysis} of our algorithms on
  Intel Skylake. Our analysis shows that the execution time of our
  algorithms indeed grows only linearly in practice, which leads to
  significant speed-up compared to the existing quadratic
  algorithms. While previous approaches rely on assumptions about the
  number of controlled physical bits (provided by huge and regular
  pages), our algorithms enable, for first time, computing eviction
  sets in scenarios without any control of the mapping from virtual
  addresses to cache sets, as
  in~\cite{SchwarzJSZero2018,Irazoqui2015}.

\end{asparaitem}

\paragraph*{Summary of contributions}
Our contributions are both theoretical and practical. On the
theoretical side, we formalize the problem of finding minimal eviction
sets and devise novel algorithms that improve the state-of-the-art
from quadratic to linear. On the practical side, we perform a rigorous
empirical analysis that exhibits the conditions under which our
algorithms succeed or fail. Overall, our insights provide a basis for
principled countermeasures against, or paths for further improving the
robustness of, algorithms for finding eviction sets.

We also include a tool for evaluating, on different platforms,
all tests and algorithms presented in this paper:
\begin{center}
\url{https://github.com/cgvwzq/evsets}
\end{center}

\section{A Primer on Caching and Virtual Memory}\label{sec:background}

In this section we provide the necessary background and notation used along the paper.

\subsection{Caches}\label{ssec:caching}

Caches are fast but small memories that bridge the latency gap between the CPU and main memory. To profit from spatial locality and to reduce management overhead, main memory is logically partitioned into a set of {\em blocks}. Each block is cached as a whole in a cache {\em line} of the same size. When accessing a block, the cache logic has to determine whether the block is stored in the cache (a cache {\em hit}) or not (a cache {\em miss}). For this purpose, caches are partitioned into equally sized {\em cache sets}. The size or number of lines of cache sets is called {\em associativity} $\assoc$ (or ways) of the cache.

\paragraph*{Cache Replacement Policies}
Since the cache is much smaller than main memory, a {\em replacement policy} must decide which memory block to evict upon a cache miss.  Traditional replacement policies include least-recently used (LRU), pseudo-LRU (PLRU), and first-in first-out (FIFO).  In modern microarchitectures, replacement policies are often more complex and generally not documented. For example, recent Intel CPUs rely on replacement policies~\cite{Jaleel2010,ReplacementPolicy2013} that dynamically adapt to the workload. 
\paragraph*{Cache Hierarchies}
Modern CPU caches are organized in multiple levels, with small and fast lower-level caches per CPU core, and a larger but slower {\em last-level cache} (LLC) that is shared among different cores. The relationship between the content of different cache levels is governed by an inclusion policy.  Intel caches, for instance, are usually {\em inclusive}. This means that the content of higher level caches (L1 and L2) is always a subset of the LLC's. In particular, blocks that are evicted from the LLC are also invalidated in higher levels. In this paper we focus on inclusive LLCs.

\paragraph*{Mapping Memory Blocks to Cache Sets}
The mapping of main memory content to the cache sets of a LLC is determined by the content's physical address. For describing this mapping, consider an architecture with $\archbits$-bit physical addresses, cache lines of $2^{\linebits}$ bytes, and $2^\indexbits$ cache sets. The least significant $\linebits$ bits of a physical address $y=(b_{\archbits-1},\dots,b_0)$ form the {\em line offset} that determines the position within a cache line. Bits $(b_{\indexbits+\linebits-1},\dots,b_\linebits)$ of $y$ are the {\em set index} bits that determine the cache set, and we denote them by $\projset{y}$.  The most significant $\archbits-\linebits-\indexbits$ bits form the {\em tag} of $y$.  See Figure~\ref{fig:bitmap} for a visualization of the role of address bits on a Intel Skylake machine.

\paragraph*{Cache Slicing}
Modern Intel CPUs partition the LLC into different $2^\slicebits$ many {\em slices}, typically one or two per CPU core. The slice is determined by an undocumented $\slicebits$-bit hash of the most significant $\archbits-\linebits$ bits of the address. With slicing, the $\indexbits$ set index bits only determine the cache set {\em within} each slice.

The total cache size $\sizeof{M}=2^{\slicebits+\indexbits+\linebits}\assoc$ is then determined as the product of the number of slices, the number of cache sets per slice, the size of each line, and the associativity.

\begin{figure}[ht]
\includegraphics[width=\linewidth]{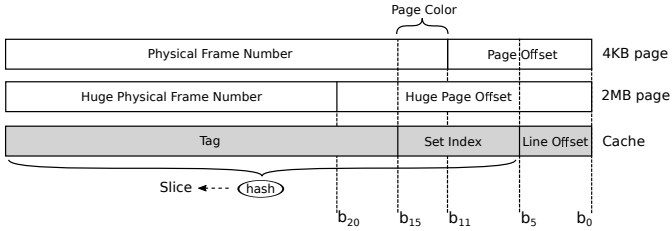}
\caption{Mapping from physical addresses to cache sets for Intel Skylake LLC, with 4 cores, $8$ slices ($\slicebits=3$), $1024$ cache sets per slice ($\indexbits=10$), lines of $64$ bytes ($\linebits=6$), and associativity $\assoc=12$.
The figure also displays page offsets and frame numbers for $4$KB pages ($\pagebits=12$) and $2$MB huge pages ($\pagebits=21$). The set index bits that
are {\em not} part of the page offset determine the {\em page color}.}\label{fig:bitmap}
\end{figure}

\subsection{Virtual Memory}\label{ssec:virtual-physical}

Virtual memory is an abstraction of the storage resources of a process that provides a linear memory space isolated from other processes and larger than the physically available resources. Operating systems, with help from the CPU's \textit{me\-mo\-ry management unit} (MMU), take care of the translation of virtual addresses to physical addresses.

\paragraph*{Virtual Address Translation}
Physical memory is partitioned in {\em pages} of size $2^\pagebits$. Common page sizes are $4$KB (i.e. $\pagebits=12$), or $2$MB for {\em huge pages}\footnote{See Appendix~\ref{apdx:hugepages} for a discussion of the availability of huge pages on different operation systems.} (i.e. $\pagebits=21$).

We model the translation from virtual to physical addresses as a function $\vptrans$ that acts as the identity on the least significant $\pagebits$ bits (named \textit{page offset}) of a virtual address $x=(x_{48},\dots,x_0)$.
That is, the virtual and physical addresses coincide on $(x_{\pagebits-1},\dots,x_0)$.  $\vptrans$ maps the most significant $48-\pagebits$ bits, named \textit{virtual page number} (VPN), to the \textit{physical frame number} (PFN). We discuss how $\vptrans$ acts on these bits in Section~\ref{subsec:probeviction}.  Figure~\ref{fig:bitmap} includes a visualization of the page offsets and physical frame numbers for small and huge pages, respectively.

\begin{figure}[ht]
\includegraphics[width=\linewidth]{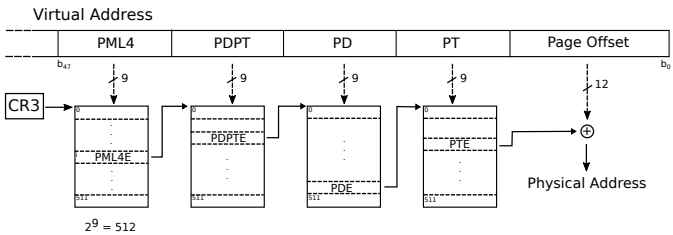}
\caption{Page walk on a 64-bit system with four levels of page tables: {\em PageMap Level 4}, {\em Page Directory Pointer}, {\em Page Directory}, and {\em Page Table} for $4$KB pages, respectively. $2$MB huge pages can be implemented by using a PD Entry directly as PT Entry. CPU's {\em Control Register 3} (CR3) points to the PML4 of the running process.}\label{fig:pgt_walk}
\end{figure}
\paragraph*{Implementing Virtual Address Translation}
Operating systems keep track of the virtual-to-physical mapping using a radix tree structure called \textit{page table} (PT) that is capable of storing the required information efficiently. Whenever a virtual address is accessed, the MMU traverses the PT until it finds the corresponding physical address. This process, also known as a \textit{page walk}, is illustrated in Figure~\ref{fig:pgt_walk}. The bits of the VPN are divided into 9-bit indexes for each level of the PT, which can store up to 512 entries (of 8 bytes each). To avoid performing a page walk for each memory access, each CPU core has a \textit{translation lookaside buffer} (TLB) that stores the most recent translations. A page walk only occurs after a TLB miss.

\section{Eviction Sets}\label{sec:eviction}

In this section we give the first formalization of eviction sets, and
present tests that enable determining whether a given set of addresses
is an eviction set. We then express the probability that a set of random
addresses forms an eviction set as a function of its size. The development
of this section forms the basis for the algorithms we develop, and for
their evaluation.

\subsection{Defining Eviction Sets}

We say that two virtual addresses $x$ and $y$ are {\em congruent},
denoted by $x\simeq y$, if they map to the same cache set.  This is
the case if and only if the set index bits $\projset{\cdot}$ and slice bits
$\projslice{\cdot}$ of their respective physical addresses
$\vptrans(x)$ and $\vptrans(y)$ coincide. That is, $x\simeq y$ if and
only if:
\begin{equation}\label{eq:sameset}
\projset{\vptrans(x)}=\projset{\vptrans(y)} \land \projslice{\vptrans(x)}=\projslice{\vptrans(y)}\
\end{equation}
Congruence is an equivalence relation. The equivalence class
$\setclass{x}$ of $x$ w.r.t. $\simeq$ is the set of virtual addresses
that maps to the same cache set as $x$. We say that addresses are {\em
  partially congruent} if they satisfy the first term of
Equation~\eqref{eq:sameset}, i.e., they coincide on the set index
bits but not necessarily on the slice bits.

We now give definitions of eviction sets, where we distinguish between
two goals: In the first, we seek to evict a specific address from the
cache. This is relevant, for example, to perform precise flushing in
rowhammer attacks.  In the second, we seek to evict the content of an
arbitrary cache set. This is relevant, for example, for high bandwidth
covert channels, where one seeks to control a large number of cache sets,
but does not care about which ones.
\begin{definition}\label{def:evictset}
  We say that a set of virtual addresses $S$ is
\begin{asparaitem}
\item an {\em eviction set for $x$} if $x\not\in S$ and at least
  $\assoc$ addresses in $S$ map to the same cache set as $x$:
\begin{equation*}\label{eq:evictsetfor}
\sizeof{\setclass{x} \cap\ S} \ge \assoc
\end{equation*}
\item an {\em eviction set (for an arbitrary address)} if there exists
  $x\in S$ such that $S\setminus\{x\}$ is an eviction set for
  $x$:
\begin{equation*}\label{eq:evictset}
\exists x:\sizeof{\setclass{x}\cap S}\ge \assoc+1
\end{equation*}
\end{asparaitem}
\end{definition}

The intuition behind Definition~\ref{def:evictset} is that
sequentially ac\-cess\-ing all elements of an eviction set for $x$ will
ensure that $x$ is {\em not} cached afterwards. Likewise, sequentially
accessing $\assoc+1$ congruent elements will guarantee that at least
one of them is being evicted.

For this intuition to hold, the cache replacement policy needs to
satisfy a condition, namely that a sequence of $\assoc$ misses to a
cache set evicts all previous content. This condition is satisfied,
for example, by all permutation-based policies~\cite{ispass04}, which
includes LRU, FIFO, and PLRU.  However, the condition is only
partially satisfied by modern (i.e. post Sandy Bridge) Intel CPUs. See
Section~\ref{sec:robustness} for a more detailed discussion.

\subsection{Testing Eviction Sets}\label{subsec:test}

Identifying eviction sets based on Definition~\ref{def:evictset} involves
checking whether~\eqref{eq:sameset} holds. This requires access to bits
of the physical addresses and cannot be performed by user programs. In
this section we present tests that rely on a timing side-channel to
determine whether a set of virtual addresses is an eviction set.
\renewcommand{\figurename}{Test}
\newcounter{tests}
\renewcommand\thefigure{\arabic{tests}}

\begin{figure}[t]
\stepcounter{tests}
\centering
\includegraphics[height=2em]{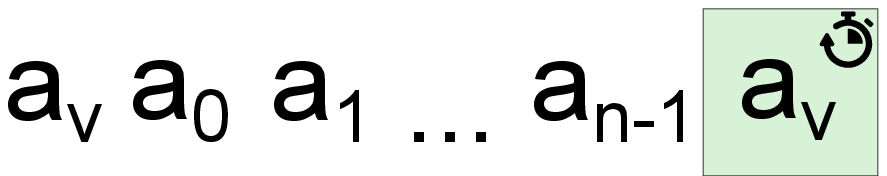}
\caption{Eviction test for a specific address $a_v$: (1) Access
  $a_v$. (2) Access $S=\{a_0,\dots,a_{n-1}\}$. (3) Access $a_v$. If the time for
  (3) is larger than a threshold, then $S$ is an eviction set
  for $a_v$.}
\label{fig:specific_test}
\end{figure}

\begin{figure}[t]
\stepcounter{tests}
\centering
\includegraphics[height=2em]{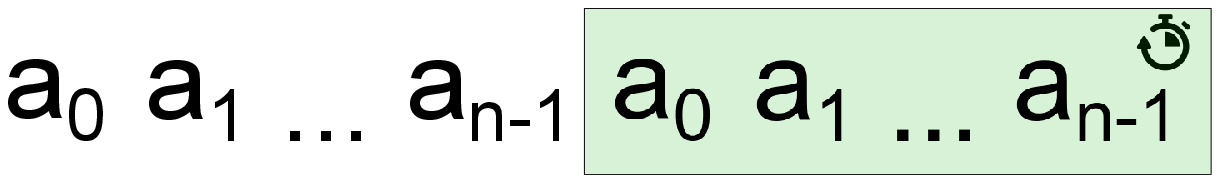}
\caption{Eviction test for an {\em arbitrary} address: (1) Access
  $S=\{a_0,\dots,a_{n-1}\}$. (2) Access $S$ again. If the overall time for
  (2) is above a threshold, $S$ is an eviction set.}
\label{fig:any_test_noisy}
\end{figure}

\begin{figure}[t]
\stepcounter{tests}
\centering
\includegraphics[height=2em]{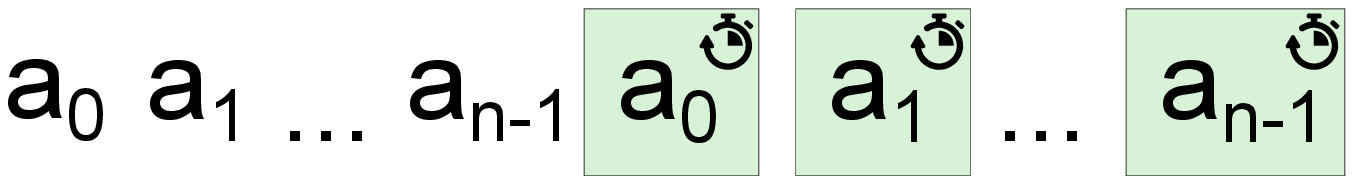}
\caption{Robust eviction test for an {\em arbitrary} address: (1) Access
  $S=\{a_0,\dots,a_{n-1}\}$. (2) Access $S$ again, measuring access time
  of each element. If the access times of more than $\assoc$
  elements in (2) is above a threshold, $S$ is an eviction set.}
\label{fig:any_test}
\end{figure}

\renewcommand{\figurename}{Fig.}
\setcounter{figure}{2}
\renewcommand\thefigure{\arabic{figure}}

\begin{asparaitem}
\item Test~\ref{fig:specific_test} from~\cite{LiuLLC2015,Irazoqui2015}
  enables user programs to check whether $S$ is an eviction set for a
  specific virtual
  address $a_v$. The test relies on the assumption that a program can
  determine whether $a_v$ is cached or not. In practice, this is
  possible whenever a program has access to a clock that enables it to
  distinguish between a cache hit or a miss.

  Test~\ref{fig:specific_test} can also be used as a basis for testing
  whether a set $S$ is an eviction set for an arbitrary address, by running
  {\sc Test}$(S\setminus\{a_i\},a_i)$, for all $a_i\in S$, and
  reporting any positive outcome. However, the number of memory
  accesses required for this is quadratic in the size of $S$.

\item Test~\ref{fig:any_test_noisy} is a more efficient solution that
  has been informally discussed in~\cite{Irazoqui2015}. The idea is to
  iterate over all the elements of $S$ twice, and measure the overall
  time of the second iteration. The first iteration ensures that all
  elements are cached. If the time for the second iteration is above a
  certain threshold, one of the elements has been evicted from the
  cache, showing that $S$ is an eviction set.  The downside of
  Test~\ref{fig:any_test_noisy} is that it is susceptible to noise, in
  the sense that any delay introduced during the second
  iteration will result in a positive answer.

\item We propose Test~\ref{fig:any_test} as a variant of
  Test~\ref{fig:any_test_noisy}, drawing inspiration from Horn's
  proof-of-concept implementation of Spectre~\cite{SpectreP0}. By
  measuring the individual time of {\em each} access instead of the
  overall access time one can (1) reduce the time window in
  which spurious events can pollute the measurements, and (2) count
  the exact number of cache misses in the second
  iteration. While this improves
  robustness to noise, it also comes with a cost in terms
  of the number of executed instructions.

\end{asparaitem}
\subsection{The Distribution of Eviction Sets}\label{subsec:probeviction}

In this section we analyze the distribution of eviction sets. More
specifically, we compute the probability that a suitably chosen set of
random virtual addresses forms an eviction set, for different degrees
of adversary control.

\subsubsection*{Choosing Candidate Sets}

For explaining what ``suitably chosen'' means, we need to distinguish
between the $\attackerbits$ set index bits of the physical addresses
that can be controlled from user space, and the
$\indexbits-\attackerbits$ bits that cannot. The value of
$\attackerbits$ depends, for example, on whether we are considering
huge or small pages.

Controlling set index bits from user space is possible because the
virtual-to-physical translation $\vptrans$ acts as the identity on the
page offset, see Section~\ref{ssec:virtual-physical}. When trying to
find a minimal eviction set, one only needs to consider virtual
addresses that coincide on those $\attackerbits$ bits.

The challenge is to find collisions on the $\indexbits-\attackerbits$
set index bits that can{\em not} be controlled from user space (the page
color bits in Figure~\ref{fig:bitmap}), as well as on the unknown
$\slicebits$ slice bits. In this section, we assume that the
virtual-to-physical translation $\vptrans$ acts as a random function
on those bits.  This assumption corresponds to the worst case from an
adversary's point of view; in reality, more detailed knowledge about
the translation can reduce the effort for finding eviction
sets~\cite{BosmanDedup2016}.

Whenever we speak about ``choosing a random set of virtual addresses''
of a given size in this paper, we hence refer to choosing random
virtual addresses that coincide on all $\attackerbits$ set index bits
under control. We now determine the probability of such a set to be
an eviction set.

\subsubsection*{Probability of Colliding Virtual Addresses}

We first compute the probability that two virtual addresses $y$ and
$x$ that coincide on the $\attackerbits$ user-controlled set index
bits are actually congruent. We call this event a {\em collision} and
denote it by $C$. As $\vptrans$ acts as a random function on the
remaining $\indexbits-\attackerbits$ set index bits and $\slicebits$
slice bits, we have:
\begin{equation*}
P(C)=2^{\attackerbits-\indexbits-\slicebits}
\end{equation*}

The following example illustrates how removing adversary control
increases the difficulty of finding collisions on common cache
configurations.

\begin{example}\label{ex:range}
  Consider the cache from Figure~\ref{fig:bitmap} with 8 slices
  (i.e. $\slicebits=3$) of $1024$ cache sets each
  (i.e. $\indexbits=10$).
\begin{asparaitem}
\item With huge pages (i.e. $\pagebits=21$), the
  attacker controls all of the set index bits,
  i.e. $\attackerbits=\indexbits$, hence the probability of a
  collision $P(C)=2^{-3}$ is given by the number of slices.
\item With pages of $4$KB (i.e. $\pagebits=12$), the number of
  bits under control is $\attackerbits=\pagebits-\linebits=6$, hence
  the probability of finding a collision is $P(C)=2^{6-10-3}=2^{-7}$.
\item The limit case (i.e. $\pagebits=\linebits=6$) corresponds to an
  adversary that has no control whatsoever over the mapping of virtual
  addresses to set index bits and slice bits -- besides the fact that
  a virtual address always maps to the same physical address. This
  case corresponds to adding a permutation layer to all
  adversary-controlled bits (see, e.g. \cite{SchwarzJSZero2018}) and
  is a candidate for a countermeasure that makes finding eviction sets
  intractable. For this case we obtain $P(C)=2^{-10-3}=2^{-13}$.
\end{asparaitem}
\end{example}

\subsubsection*{Probability of a Set to be an Eviction Set for $x$}

We analyze the probability of a set of virtual addresses $S$ to be
an eviction set for a given address $x$. This probability can be expressed in
terms of a binomially distributed random variable $X \sim B(N, p)$
with parameters $N=\sizeof{S}$ and $p=P(C)$.
With such an $X$, the probability of finding $k$ collisions, i.e., $\sizeof{S\cap\setclass{x}}=k$, is given by:
\begin{equation*}
P(X=k)=\binom{N}{k}p^k(1-p)^{N-k}
\end{equation*}
According to Definition~\ref{def:evictset}, $S$ is an eviction set if it contains at least $\assoc$ addresses that are congruent with $x$, see~\eqref{eq:evictsetfor}. The probability of this happening is given by:
\begin{align*}
P(\sizeof{S\cap \setclass{x}}\ge \assoc) & = 1 - P(X < \assoc) \\
& = 1 - \sum_{k=0}^{\assoc-1} \binom{N}{k}p^k(1-p)^{N-k}
\end{align*}

Figure~\ref{fig:probabilities} depicts the distribution of sets to be
an eviction set for $x$, based on the cache from Figure~\ref{fig:bitmap}.

\subsubsection*{Probability of a Set to be an Eviction Set for an arbitrary address}
We analyze the probability that a set $S$ contains at least $\assoc+1$
addresses that map to the same cache set. To this end, we view the
problem as a cell occupancy problem.

Namely, we consider
$B=2^{\slicebits+\indexbits-\attackerbits}$ possible cache
sets (or bins) with $N=\sizeof{S}$ addresses (or balls) that are
uniformly distributed, and ask for the probability of filling at least
one set (or bin) with more than $\assoc$ addresses (or balls).

We model this probability using random variables $N_1,\dots,N_B$,
where $N_i$ represent the number of addresses mapping to the $i$-th
cache set, with the constraint that $N = N_1 + ... + N_B$. With this,
the probability of having at least one set with more than $\assoc$
addresses can be reduced to the complementary event of all $N_i$
being less or equal than $\assoc$:

\begin{equation*}
P(\exists i \mid N_i > \assoc) = 1 - P(N_1\leq \assoc,...,N_B\leq
\assoc)
\end{equation*}
The right-hand side is a cumulative multinomial distribution, whose
exact combinatorial analysis is expensive for large values of $N$ and
becomes unpractical for our purpose. Instead, we rely on a well-known
approximation based on Poisson distributions~\cite{levin1981} for calculating
the probabilities.

Figure~\ref{fig:probabilities} depicts the distribution of sets to be
an eviction set for an arbitrary address, based on the cache from
Figure~\ref{fig:bitmap}.  We observe that the probability of the
multinomial grows faster with the set size than the binomial
distribution. This shows that a set is more likely an eviction set for
an arbitrary address than it is for a specific address.

\begin{figure}[ht]
\begin{tikzpicture}[baseline=(current axis.outer east)]
	\begin{axis}[
        width=\columnwidth,
        height=0.4\columnwidth,
        try min ticks=5,
        xlabel=Set Size,
	]
	\addplot[smooth,mark size=1pt,color=red,mark=x] coordinates {
        (0,0.0)
        (492,0.001)
        (696,0.01)
        (888,0.05)
		(1008,0.1)
		(1164,0.2)
		(1284,0.3)
		(1392,0.4)
		(1500,0.5)
		(1608,0.6)
		(1740,0.7)
		(1896,0.8)
		(2124,0.9)
		(3420,1.0)
        (4000,1.0)
	};
	\addplot[smooth,mark size=0.8pt,color=blue,mark=*] coordinates {
        (0,0.0)
        (348,0.001)
        (456,0.011)
        (540,0.053)
		(588,0.1)
		(636,0.2)
		(672,0.3)
		(696,0.4)
		(732,0.5)
		(756,0.6)
		(780,0.7)
		(804,0.8)
		(840,0.9)
		(984,1.0)
        (1000,1.0)
        (2000,1.0)
        (3000,1.0)
        (4000,1.0)
	};
	\end{axis}
\end{tikzpicture}
\caption{Probability of random sets to be eviction sets as a function of their size, based on our theoretical models. We use $P(C)=2^{-7}$ to represent an attacker with $4$KB pages in the machine from Figure~\ref{fig:bitmap}. The {\em blue-circle} line shows the multinomial model for an ``arbitrary'' eviction set. The {\em red-cross} line shows the binomial model for an ``specific'' eviction set.}
\label{fig:probabilities}
\end{figure}
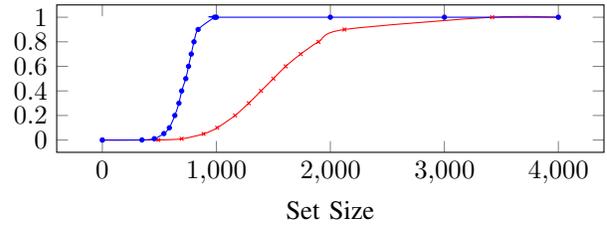

\subsubsection*{Cost of Finding Eviction Sets}

We conclude this section by computing the cost (in terms of
the expected number of memory accesses required) of finding an eviction set of
size $N$ by repeatedly and independently selecting and testing candidate
sets.

To this end, we model the repeated independent choice of eviction sets
as a geometric distribution over the probability $p(N)$ that a
candidate set of size $N$ is an eviction set. The expectation $1/p(N)$ of
this distribution captures the expected number of candidate sets that
must be tested until we find an eviction set. Assuming that a test of
a set of size $N$ requires $\oo(N)$ memory accesses, as in
Section~\ref{subsec:test}, this yields an overall cost in terms of
memory accesses for finding an initial eviction set of
$\oo(N/p(N))$.

\begin{figure}[H]
\includegraphics[width=.9\columnwidth]{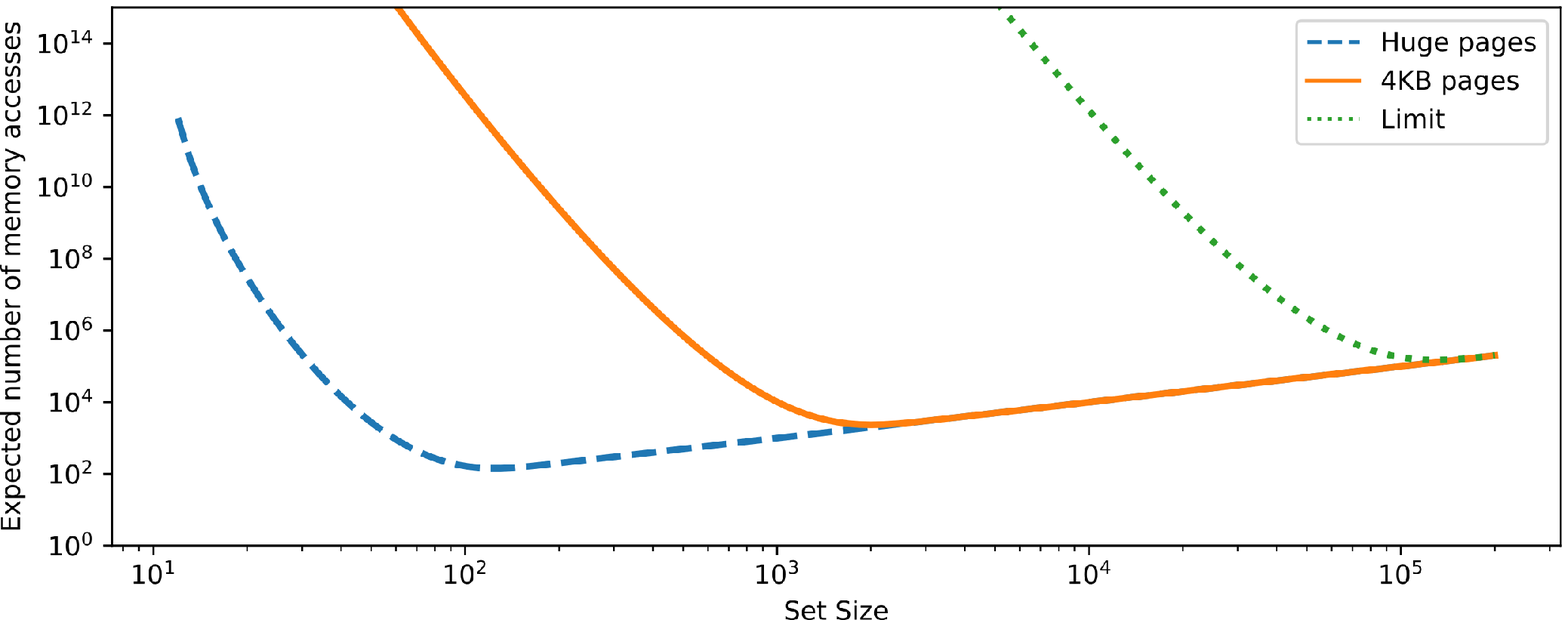}
\caption{Expected number of memory accesses for finding an eviction set as
a function of its size. The {\em dashed blue} line represents $P(C)=2^{-3}$,
an attacker with huge pages (i.e. controls all $\attackerbits=10$ set index
bits). The {\em plain orange} line represents $P(C)=2^{-7}$, an attacker
with $4$KB pages (i.e. controls $\attackerbits=6$). The {\em dotted
  green} line
represents $P(C)=2^{-13}$, an attacker w/o any control over the set index
bits (i.e. $\attackerbits=0$).}
\label{fig:cost_per_size}
\end{figure}

Figure~\ref{fig:cost_per_size} depicts the cost function $N/p(N)$ for
the adversaries from Example~\ref{ex:range}) for finding eviction sets
for a specific address, and highlights the most favorable sizes for
finding eviction sets. Since probability grows with set size, finding
an eviction set of small size requires, in expectation, large number
of trials. Once the probability stabilizes (i.e. the set is large
enough), we start seeing the linear cost of the test.

\section{Algorithms for Computing Minimal Eviction Sets}\label{sec:algorithms}

The probability that a set of virtual addresses forms an eviction set depends on its size, on the cache settings (e.g., associativity and number of cache sets), and on the amount of control an adversary has over the physical addresses. In particular, a small set of random virtual addresses is unlikely to be an eviction set. This motivates the two-step approach for finding minimal eviction sets in which one (1) first identifies a large eviction set, and (2) then reduces this set to its minimal core.

Previous proposals in the literature rely on this two-step approach. In this section we first present the baseline reduction from the literature, which requires $\oo(N^2)$ memory accesses. We then show that it is possible to perform the reduction using only $\oo(N)$ memory accesses, which enables dealing with much larger initial eviction sets than before.

The main practical implication of this result is that finding minimal eviction sets from user (or sandboxed) space is faster than previously thought, and hence practical even without any control over the slice or set index bits. This renders countermeasures based on reducing adversary control over these bits futile.

\subsection{The Baseline Algorithm}

We revisit the baseline algorithm for reducing eviction sets that has been informally described in the literature. Its pseudocode is given as Algorithm~\ref{alg:naivetesting}.

Algorithm~\ref{alg:naivetesting} receives as input a virtual address $x$ and an eviction set $S$ for $x$. It proceeds by picking an address $c$ from $S$ and tests whether $S\setminus\{c\}$ is still evicting $x$, see line~\ref{l:check}. If it is {\em not} (notably the {\em if}-branch), $c$ must be congruent to $x$ and is recorded in $R$, see line~\ref{l:record}. The algorithm then removes $c$ from $S$ in line~\ref{l:remove} and loops in line~\ref{l:recurse}.

 Note that the eviction test {\sc Test} is applied to $R\cup(S\setminus\{c\})$ in line~\ref{l:check}, i.e. all congruent elements found so far are included. This enables scanning $S$ for congruent elements even when there are less than $\assoc$ of them left. The algorithm terminates when $R$ forms a minimal eviction set of $\assoc$ elements, which is guaranteed to happen because $S$ is initially an eviction set.
\begin{proposition}\label{prop:naive}
Algorithm~\ref{alg:naivetesting} reduces an eviction set $S$ to its minimal core in $\oo(N^2)$ memory accesses, where $N = \sizeof{S}$.
\end{proposition}
The complexity bound follows because $\sizeof{S}$ is an upper bound for the number of loop iterations as well as on the argument size of Test~\ref{alg:naivetesting} and hence the number of memory accesses performed during each call to {\sc test}.

\begin{algorithm}[tb]
\caption{Baseline Reduction}\label{alg:naivetesting}
\hspace*{\algorithmicindent} \textbf{In:} $S$=candidate set, $x$=victim address \\
\hspace*{\algorithmicindent} \textbf{Out:} $R$=minimal eviction set for $v$\\
\begin{algorithmic}[1]
\State {$R \leftarrow \{\}$}
\While{$|R| < \assoc$}\label{l:recurse}
\State $c\leftarrow$ \Call{$pick$}{$S$}\label{l:pick}
\If{$\neg$\Call{test}{$R \cup (S\setminus\{c\})$, $x$}}\label{l:check}
\State {$R \leftarrow R \cup \{c\}$}\label{l:record}
\EndIf
\State $S \leftarrow S\setminus\{c\}$\label{l:remove}
\EndWhile
\State \Return $R$
\end{algorithmic}
\end{algorithm}

The literature contains different variants of
Algorithm~\ref{alg:naivetesting}\cite{LiuLLC2015, Irazoqui2015, OrenSpy2015}.
For example, the variant presented in~\cite{OrenSpy2015} always puts $c$ back
into $S$ and keeps iterating until $\sizeof{S} = \assoc$. This still is
asymptotically quadratic, but adds some redundancy that helps to combat errors.

If the quadratic baseline was optimal, one could think about preventing
an adversary from computing small eviction sets by reducing or removing
control over the set index bits, either by adding a mask via
hardware~\cite{Irazoqui2015} or a permutation layer via
software~\cite{SchwarzJSZero2018} (see limit case in Example~\ref{ex:range}).

\subsection{Computing Minimal Eviction Sets for a Specific Address}

We present a novel algorithm that performs the reduction of eviction sets to their minimal core in $\mathcal{O}(N)$ memory accesses, where $N = \sizeof{S}$. This enables dealing with much larger eviction sets than the quadratic baseline and renders countermeasures based on hiding the physical address futile. Our algorithm is based on ideas from threshold group testing, which we briefly introduce first.

\subsubsection{Threshold Group Testing}\label{ssec:grouptesting}

Group testing refers to procedures that break up the task of identifying elements with a desired property by tests on sets (i.e. {\em groups} of those elements). Group testing has been proposed for identifying diseases via blood tests where, by testing {\em pools} of blood samples rather than individual samples, one can reduce the number of tests required to find all positive individuals from linear to logarithmic in the population size~\cite{Dorfman1943}.

{\em Threshold} group testing~\cite{Damaschke2006} refers to group testing based on
 tests that give a negative answer if the number of positive individuals in the tested set is at most $\lowth$, a positive answer if the number is at least $\upth$, and any answer if it is in-between $\lowth$ and $\upth$.
Here, $\lowth$ and $\upth$ are natural numbers that represent lower and upper thresholds, respectively.

\subsubsection{A Linear-Time Algorithm for Computing Minimal Eviction Sets}\label{sec:lintime}

The key insight behind our algorithm is that testing whether a set of virtual addresses $S$ evicts $x$ (see Test~\ref{fig:specific_test}) can actually be seen as a threshold group test for congruence with $x$, where $\lowth=\assoc-1$ and $\upth=\assoc$. This is because the test gives a positive answer if $\sizeof{\setclass{x} \cap S} \ge \assoc$, and a negative answer otherwise. This connection allows us to leverage the following result from the group testing literature for computing minimal eviction sets.

\begin{lemma}[\cite{Damaschke2006}]\label{lem:numtests}
  If a set $\univset$ contains $\genassoc$ or more positive elements, one can identify $\genassoc$ of them using $\oo{(\genassoc \log \sizeof{\univset})}$ threshold
  group tests with $\lowth=\genassoc-1$ and $\upth=\genassoc$.
\end{lemma}
\begin{proof}
 The idea behind Lemma~\ref{lem:numtests} is to partition $S$ in $\genassoc+1$ disjoint subsets $T_1,\dots,T_{\genassoc+1}$ of (approximately) the same size. A counting argument (see Appendix~\ref{apdx:pidgeonhole}) shows that there is at least one $j\in\{1,\dots,\genassoc+1\}$ such that $S \setminus T_{j}$ is still an eviction set. One identifies such a $j$ by group tests and repeats the procedure on $S\setminus T_{j}$. The logarithmic complexity is due to the fact that $\sizeof{S\setminus T_{j}}=\sizeof{S}\frac{\genassoc}{\genassoc+1}$, i.e. each iteration reduces the eviction set by a factor of its size,
 rather by than a constant as in Algorithm~\ref{alg:naivetesting}.
\end{proof}%

\begin{algorithm}[tb]
\caption{Reduction Via Group Testing}\label{alg:grouptesting}
\hspace*{\algorithmicindent} \textbf{In :} $S$=candidate set, $x$=victim address\\
\hspace*{\algorithmicindent} \textbf{Out :} $R$=minimal eviction set for $x$ \\
\begin{algorithmic}[1]
  \While{$|S| > \assoc$}
	\State {$\{T_{1}, ..., T_{\assoc+1}\} \leftarrow$} \Call{$split$}{$S$, $\assoc+1$}
	\State {$i \leftarrow 1$}
	\While{$\neg$\Call{test}{$S \setminus T_{i}$, $x$}}
		\State {$i \leftarrow i + 1$}
	\EndWhile
	\State {$S \leftarrow S \setminus T_{i}$}
\EndWhile
\State \Return $S$
\end{algorithmic}
\end{algorithm}

Algorithm~\ref{alg:grouptesting} computes minimal eviction sets based on this idea. Note that Lemma~\ref{lem:numtests} gives a bound on the number of group tests. For computing eviction sets, however, the relevant complexity measure is the total number of memory accesses made, i.e. the sum of the sizes of the sets on which tests are performed. We next show that, with this complexity measure, Algorithm~\ref{alg:grouptesting} is linear in the size of the initial eviction set.

\begin{proposition}\label{prop:numacc}
  Algorithm~\ref{alg:grouptesting} with Test 1 reduces an eviction set $S$ to its minimal core using $O(\assoc^2 N)$ memory accesses, where $N = \sizeof{S}$.
\end{proposition}
\begin{proof}
  The correctness of Algorithm~\ref{alg:grouptesting} follows from the invariant that $S$ is an eviction set and that it satisfies $\sizeof{S}=\assoc$ upon termination, see Lemma~\ref{lem:numtests}. For the proof of the complexity bound observe that the number of memory accesses performed by Algorithm~\ref{alg:grouptesting} on a set $S$ of size $N$ follows the following recurrence.
\begin{equation}\label{eq:grouptesting}
T(N)=T(N\frac{\assoc}{\assoc+1}) + N\cdot\assoc
\end{equation}
for $N > \assoc$, and $T(\assoc)=\assoc$.
The recurrence holds because, on input $S$, the algorithm applies threshold group tests on $\assoc+1$ subsets of $S$, each of size $N-\frac{N}{\assoc+1}$. The overall cost for the split and the tests is $N \cdot\assoc$.  The algorithm recurses on exactly one of these subsets of $S$, which has size $N\frac{\assoc}{\assoc+1}$. From the Master theorem~\cite{cormen5introduction} it follows that $T(N)\in\Theta(N)$.
\end{proof}
See Appendix~\ref{apdx:proofcost} for a direct proof that also includes the quadratic dependency on associativity.

\subsection{Computing Minimal Eviction Set for an Arbitrary Address}

The Algorithms presented so far compute minimal eviction sets for a specific virtual address $x$. We now consider the case of computing minimal eviction sets for an arbitrary address. This case is interesting because, as shown in Section~\ref{subsec:probeviction}, a set of virtual addresses is more likely to evict any arbitrary address than a specific one. That is, in scenarios where the target address is not relevant, one can start the reduction with smaller candidate sets.

The key observation is that both Algorithm~\ref{alg:naivetesting} and Algorithm~\ref{alg:grouptesting} can be easily adapted to compute eviction sets for an arbitrary address. This only requires replacing the eviction test for a specific address (Test~\ref{fig:specific_test}) by an eviction test for an arbitrary address (Test~\ref{fig:any_test}).

\begin{proposition}
Algorithm~\ref{alg:naivetesting}, with Test~\ref{fig:any_test} for an {\em arbitrary} eviction set, reduces an eviction set to its minimal core in $\oo(N^2)$ memory accesses, where $N = \sizeof{S}$.
\end{proposition}

\begin{proposition}
Algorithm~\ref{alg:grouptesting} with Test~\ref{fig:any_test} reduces an eviction set to its minimal core in $\oo(N)$ memory accesses, where $N = \sizeof{S}$.
\end{proposition}
The complexity bounds for computing eviction sets for an arbitrary address coincide with those in Proposition~\ref{prop:naive} and~\ref{prop:numacc} because Test~\ref{fig:specific_test} and Test~\ref{fig:any_test} are both linear in the size of the tested set.

\subsection{Computing Minimal Eviction Sets for {\em Many} Virtual Addresses}\label{ssec:evset4all}

We now discuss the case of finding eviction sets for a large number of cache sets. For this we assume a given pool $P$ of virtual addresses, and explain how to compute minimal eviction sets for all the eviction sets that are contained in $P$. For large enough $P$ the result can be a set of eviction sets for all virtual addresses.

The core idea is to use a large enough subset of $P$ and reduce it to a minimal eviction set $S$ for an arbitrary address, say $x$. Use $S$ to build a test {\sc Test}$((S\setminus\{x\})\cup\{y\},x)$ for individual addresses $y$ to be congruent with $x$. Use this test to scan $P$ and remove all elements that are congruent with $x$. Repeat the procedure until no more eviction sets are found in $P$. With a linear reduction using Algorithm~\ref{alg:grouptesting}, a linear scan, and a constant number of cache sets, this procedure requires $\mathcal{O}(\sizeof{P})$ memory accesses to identify all eviction sets in $P$.

Previous work~\cite{OrenSpy2015} proposes a similar approach based on the
quadratic baseline reduction. The authors leverage the fact that,
on earlier Intel CPUs, given two congruent physical addresses $x \simeq y$,
then $x + \Delta \simeq y + \Delta$, for any offset
$\Delta < 2^{\attackerbits}$. This implies that, given one eviction set for
each of the $2^{\indexbits-\attackerbits}$ page colors, one can immediately
obtain $2^{\attackerbits}-1$ others by adding appropriate offsets to each
address. Unfortunately, with unknown slicing functions this only holds with
probability $2^{-\slicebits}$, what increases the attacker's effort. Our
linear-time algorithm helps scaling to large numbers of eviction sets under
those conditions.

Another solution to the problem of finding many eviction sets has been proposed in~\cite{LiuLLC2015}. This solution differs from the two-step approach in that the algorithm first constructs a so-called {\em conflict set}, which is the union of all minimal eviction sets contained in $P$, before performing a split into the individual minimal eviction sets. The main advantage of using conflict sets is that, once a minimal eviction set is found, the conflict set need not be scanned for further congruent addresses.

\section{Evaluation}\label{sec:evaluation}

In this section we perform an evaluation of the algorithms for computing minimal eviction sets we have developed in Section~\ref{sec:algorithms}. The evaluation complements our theoretical analysis along two dimensions:
\paragraph*{Robustness} The theoretical analysis assumes that tests for eviction sets always return the correct answer, which results in provably correct reduction algorithms. In this section we analyze the robustness of our algorithms in practice. In particular, we study the influence of factors that are outside of our model, such as adaptive cache replacement policies and TLB activity. We identify conditions under which our algorithms are almost perfectly reliable, as well as conditions under which their reliability degrades. These insights can be the basis of principled countermeasures against, or paths forward for improving robustness of, algorithms for finding eviction sets.
\paragraph*{Execution time} The theoretical analysis captures the performance of our algorithms in terms of the number of memory accesses. As for the case of correctness, the real execution time is influenced by factors that are outside of our model, such as the total number of cache and TLB misses, or the implementation details. In our empirical analysis we show that the number of memory accesses is in fact a good predictor for the asymptotic real-time performance of our algorithms.

\subsection{Design of our Analysis}\label{subsec:design}

\paragraph*{Implementation}
We implemented the tests and algorithms described in Sections~\ref{subsec:test}~and~\ref{sec:algorithms} as a command line tool, which can  be parameterized to find minimal eviction sets on different platforms.
All of our experiments are performed using the tool. The source code is available at: \url{https://github.com/cgvwzq/evsets}.

\paragraph*{Analyzed Platforms}
We evaluate our algorithms on two different CPUs running Linux~4.9:
\begin{asparaenum}
\item Intel i5-6500 4 x 3.20\,GHz (Skylake family), 16\,GB of RAM, and a 6\,MB LLC with 8192 12-way cache sets. Our experiments indicate that only 10 bits are used as set index on this machine, we hence conclude that each core has 2 slices. Following our previous notation, i.e.: $\assoc_{\skylake}=12, \indexbits_{\skylake}=10, \slicebits_{\skylake}=3, \linebits_{\skylake}=6$.
\item Intel i7-4790 8 x 3.60GHz\,GHz (Haswell family), 8\,GB of RAM, and a 8\,MB LLC with 8192 16-way cache sets. This machine has 4 physical cores and 4 slices. Following our previous notation, i.e.: $\assoc_{\haswell}=16, \indexbits_{\haswell}=11, \slicebits_{\haswell}=2, \linebits_{\haswell}=6$.
\end{asparaenum}
We emphasize that all experiments run on machines with user operating systems (with a default window manager and background services), default kernel, and default BIOS settings.

\paragraph*{Selection of Initial Search Space}
We first allocate a big memory buffer as a pool of addresses from where we can
suitably chose the candidate sets (recall Section~\ref{subsec:probeviction}).
This choice is done based on the adversary's capabilities (i.e., $\attackerbits$),
for example, by collecting all addresses in the buffer using a stride of
$2^{\attackerbits + \linebits}$, and then randomly selecting $N$ of them. With
this method, we are able to simulate any amount of adversary control over the set index bits, i.e. any $\attackerbits$ with
$\attackerbits < \pagebits-\linebits$.

\paragraph*{Isolating and Mitigating Interferences}\label{ssec:isolation}
We identify ways to isolate two important sources of interference that affect the
reliability of our tests and hence the correctness of our algorithms:
\begin{asparaitem}
\item {\em Adaptive Replacement Policies:}
Both Skylake and Haswell employ mechanisms to adaptively switch between undocumented cache replacement policies.
Our experiments indicate that Skylake keeps a few {\em fixed} cache sets (for example, the cache set {\em zero}) that seem to behave as PLRU and match the assumptions of our model. Targeting such sets allows us to isolate the effect of adaptive and unknown replacement policies on the reliability of our algorithms.
\item {\em Translation Lookaside Buffers:}
Performing virtual memory translations during a test results in accesses to the TLB. An increased number of translations can lead to an increased number of TLB misses, which at the end trigger page walks. These page walks result in {\em implicit} memory accesses that may evict the target address from the cache, even though the set under test is {\em not} an eviction set, i.e. it introduces a false positive. TLB misses also introduce a noticeable delay on time measurements, what has been recently discussed in a concurrent work~\cite{CachePortableGenkin2018}. We isolate these effects by performing experiments for pages of $4$KB on huge pages of $2$MB, but under the assumption that, as for $4$KB pages, only $\attackerbits=6$ bits of the set index are under attacker control.
\end{asparaitem}

We further rely on common techniques from the literature to mitigate the influence of other sources of interference:
\begin{asparaitem}
\item For reducing the effect of {\em hardware prefetching} we use a linked list to represent eviction sets, where each element is a pointer to the next address. This ensure that all memory accesses loads are executed in-order. We further randomize the order of elements.
\item For reducing the effect of {\em jitter}, we perform several time measurements per test and compare their average value with a threshold. In our experiments, $10-50$ measurements are sufficient to reduce the interference of context switches and other spurious events. More noisy environments (e.g. a web browser) may require larger numbers.
\end{asparaitem}

\begin{figure*}[ht]
    \includegraphics[width=0.9\linewidth]{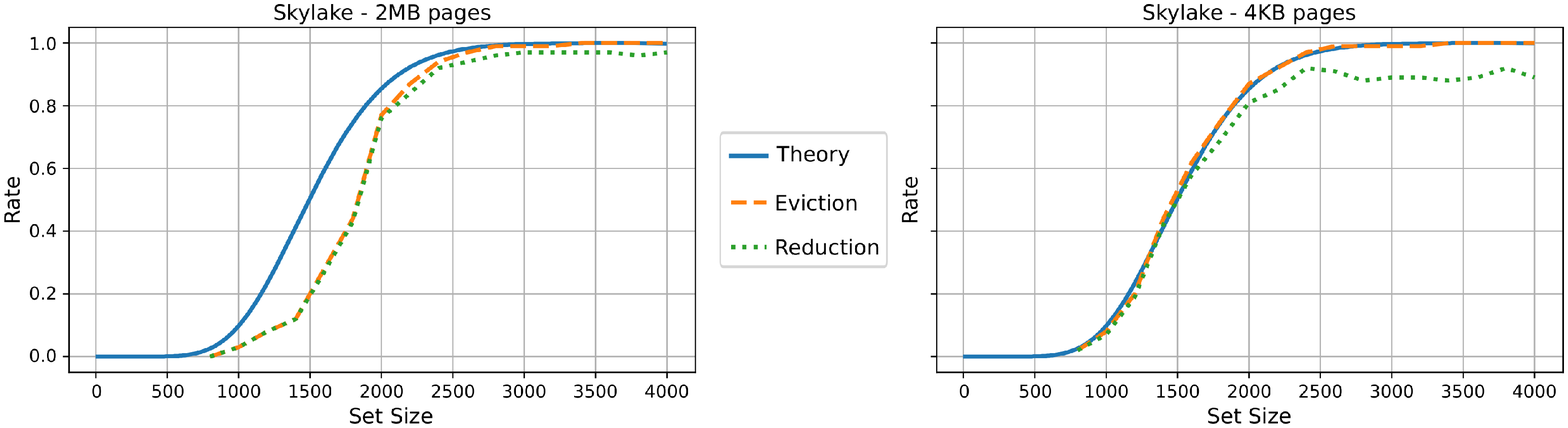}
    \caption{Skylake: Eviction for specific address $x$ on {\bf cache set {\em zero}}, compared to our binomial model. Each point is the average of $1000$ reductions for sets of $N$ randomly chosen addresses.}\label{fig:test_binomial_skylake}
\end{figure*}
\begin{figure}[ht]
    \begin{subfigure}{\columnwidth}
        \includegraphics[width=\textwidth]{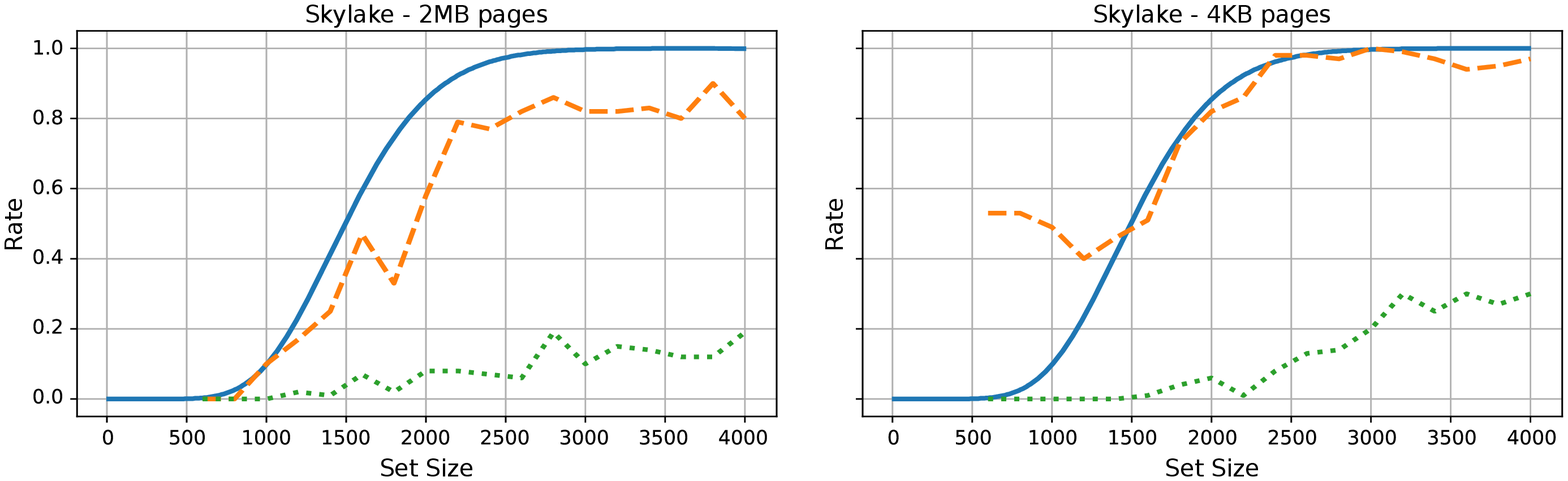}
        \caption{Eviction and reduction rates for an specific address $x$ targeting {\bf cache set 10}, compared to our binomial model.}\label{fig:test_binomial_skylake_set10}
    \end{subfigure}
    \begin{subfigure}{\columnwidth}
        \includegraphics[width=\textwidth]{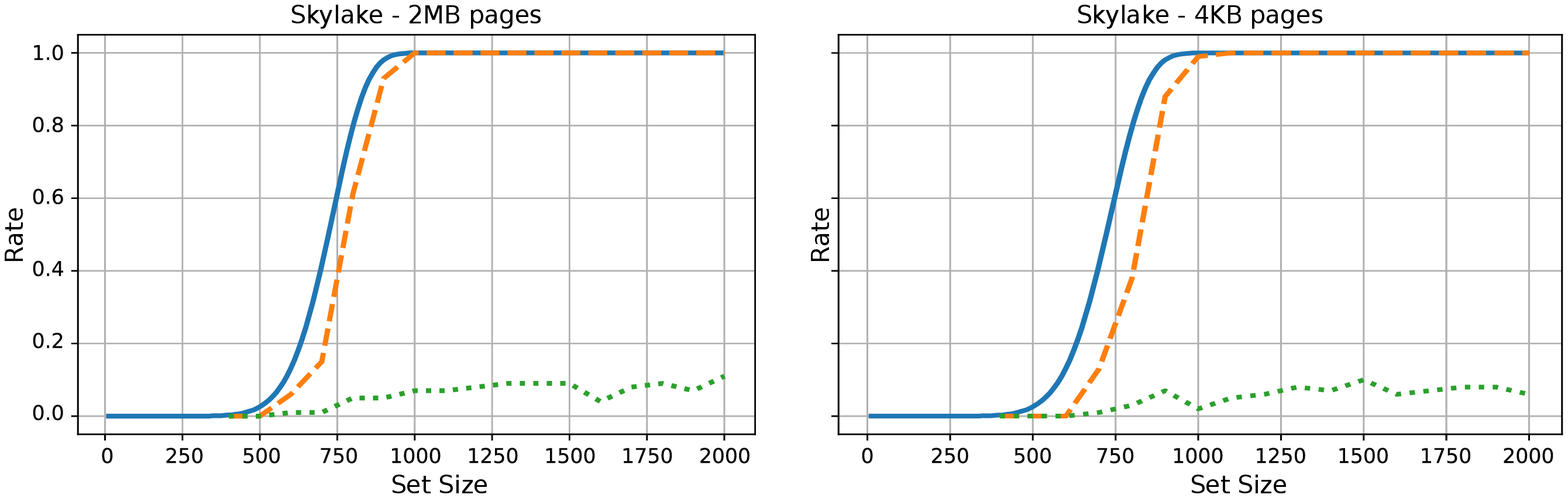}
        \caption{Eviction and reduction rates for an arbitrary address, compared to our multinomial model.}\label{fig:test_multinomial_skylake}
    \end{subfigure}
    \caption{Experiments on Skylake. Each point is the average of $100$ reductions for sets of $N$ randomly chosen addresses.}
\end{figure}
\begin{figure}[ht]
    \begin{subfigure}{\columnwidth}
        \includegraphics[width=\textwidth]{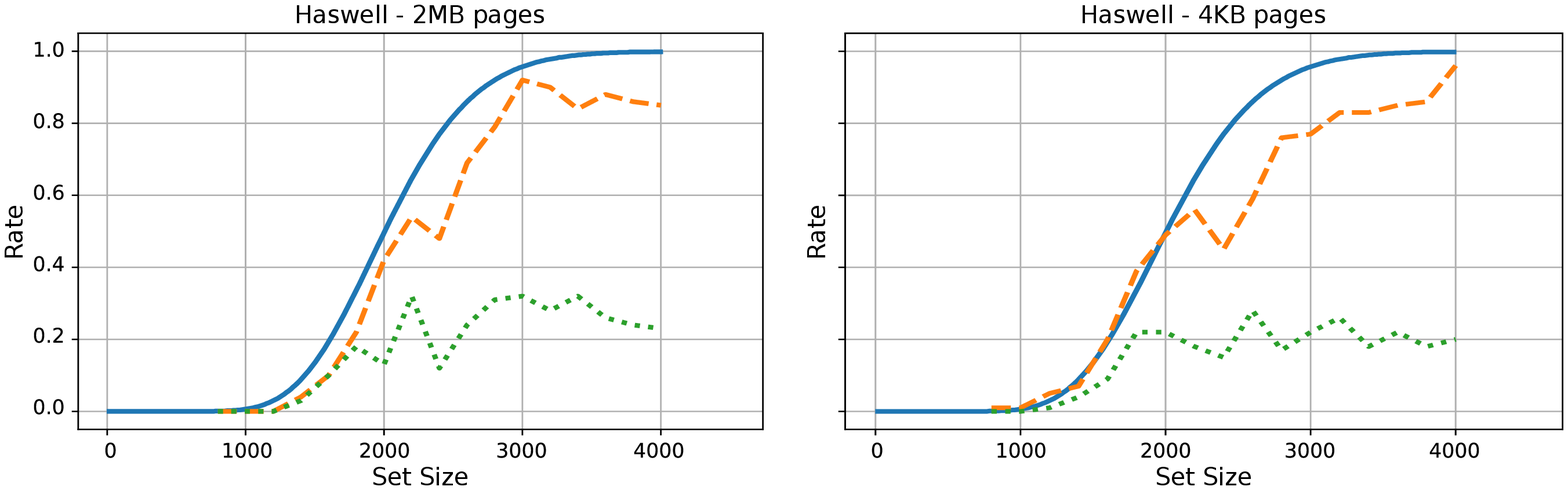}
        \caption{Eviction and reduction rates for an specific address $x$, compared to our binomial model.}\label{fig:test_binomial_haswell}
    \end{subfigure}
    \begin{subfigure}{\columnwidth}
        \includegraphics[width=\textwidth]{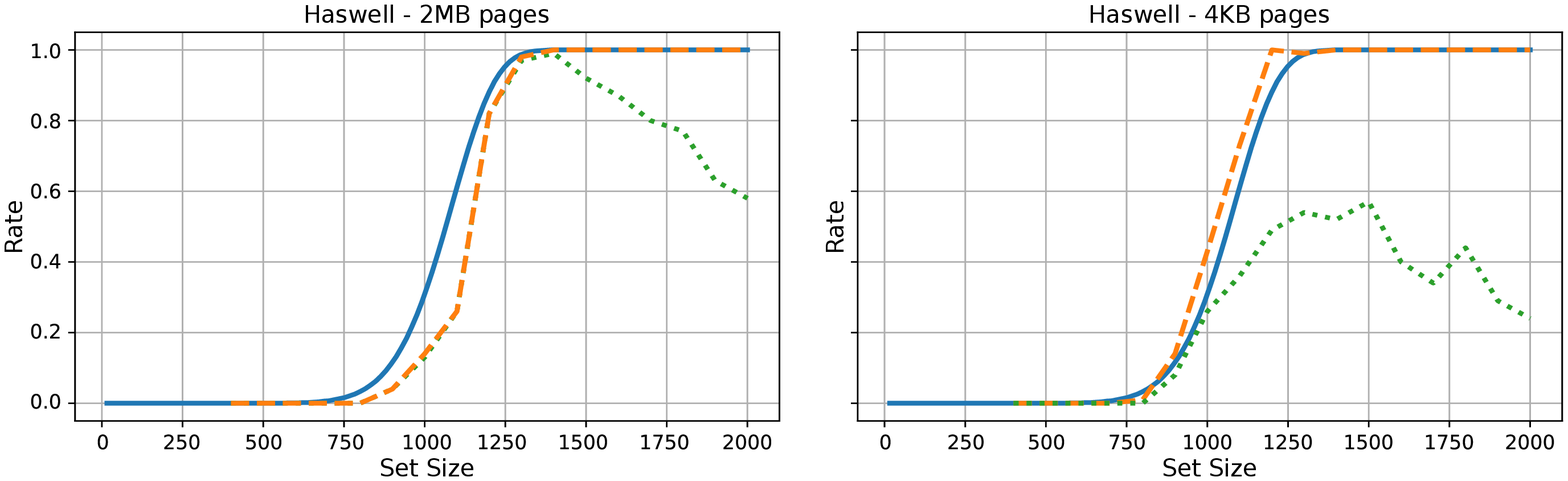}
        \caption{Eviction and reduction rates for an arbitrary address, compared to our multinomial model.}\label{fig:test_multinomial_haswell}
    \end{subfigure}
    \caption{Experiments on Haswell. Each point is the average of $100$ reductions for sets of $N$ randomly chosen addresses.}
\end{figure}

\subsection{Evaluating Robustness}\label{ssec:criteria}

We rely on two indicators for the robustness of our tests and reduction algorithms:

\begin{compactitem}
\item The {\em eviction rate}, which is the relative frequency of our tests returning true on randomly selected sets of fixed size.
\item The {\em reduction rate}, which we define as the relative frequency of our reduction succeeding to reduce randomly selected sets of fixed size to a minimal eviction set.
\end{compactitem}
Here, a reduction is successful if the elements it returns are congruent, i.e., they coincide on the set bits and on the slice bits. For this check we rely on the reverse engineered slice function for Intel CPUs~\cite{Maurice2015RAID}.

With perfect tests (and hence correct algorithms), both the eviction rate and the reduction rate should coincide with the theoretical prediction given in Section~\ref{sec:eviction}. Our analysis hence focuses on deviations of the eviction and reduction rate from these predictions.

\paragraph*{Experimental Results}

The experimental results for eviction and reduction for a specific address $x$ are given in Figures~\ref{fig:test_binomial_skylake} and~\ref{fig:test_binomial_skylake_set10} (for Skylake), and Figure~\ref{fig:test_binomial_haswell} (Haswell). The results for arbitrary addresses are given in Figures~\ref{fig:test_multinomial_skylake} and~\ref{fig:test_multinomial_haswell} (for Skylake and Haswell, respectively). We highlight the following findings:

\begin{asparaitem}

\item {\em Analysis under idealized conditions}
We first analyze our test and reduction algorithms under idealized conditions, where we use the techniques described in Section~\ref{ssec:isolation} to mitigate the effect of TLBs, complex replacement policies, prefetching, and jitter.
Figure~\ref{fig:test_binomial_skylake} illustrates that, under these conditions, eviction and reduction rates (Test~\ref{fig:specific_test} and~Algorithm~\ref{alg:grouptesting}) closely match. Moreover, eviction and reduction rates closely match the theoretical prediction for small pages.

For huge pages, however, eviction and reduction rates remain below the theoretical prediction, see Figure~\ref{fig:test_binomial_skylake}.  We attribute this to the fact that, using explicit allocations (see Appendix~\ref{apdx:hugepages}), huge pages are chosen from a pre-allocated pool of physical pages that usually resides in a fixed {\em zone}. This limits the uniformity of the more significant bits and deviates from our uniform modeling.

\item {\em Effect of the cache replacement policies}.
  Our experimental results show that the eviction and reduction rates decrease significantly on Haswell (Figure~\ref{fig:test_binomial_haswell}), and on Skylake when targeting a cache set (Figure~\ref{fig:test_binomial_skylake_set10}) different from zero. The effect is also visible in the evaluation of algorithms for finding an arbitrary eviction set (see Figures~\ref{fig:test_multinomial_skylake} and~\ref{fig:test_multinomial_haswell}).

The decrease seems to be caused by two factors: the replacement policy of the targeted cache sets does not match our models; and the targeted cache set are influenced by accesses to other sets in the cache. We provide further evidence of this effect in Section~\ref{sec:robustness}.

\item {\em Effect of TLB thrashing}.
Virtual memory translations are more frequent with small pages than with huge pages, which shows in our experiments: The eviction rate lies above the theoretical prediction, in particular for large sets, which shows the existence of false positives. In contrast, the reduction rate falls off. This is because false positives in tests cause the reduction to select sets that are {\em not} eviction sets, which leads to failure further down the path.

 The effect is clearly visible in Figure~\ref{fig:test_binomial_skylake}, where we  compare the results on small pages with those on huge pages for cache set zero on Skylake. We observe that the reduction rate on small pages declines for $N > 1500$, which, as Appendix~\ref{apdx:tlbinfo} shows, coincides with the TLB capacity of Skylake of $1536$ entries.

  The effect is also visible in Figure~\ref{fig:test_multinomial_haswell} , where we attribute the strong decline of the reduction rate after $N>1000$ (Haswell's TLB  capacity is $1024$ entries) to implicit memory accesses having a greater chance to be an eviction set for Haswell's adaptive replacement policy. In the rest of figures the effect is overlaid with interferences of the replacement policy. However, Figure~\ref{fig:test_multinomial_skylake} shows that with large TLBs, and for most reasonable values of $N$, the effect of TLB thrashing is negligible.
\end{asparaitem}

\subsection{Evaluating Performance}\label{ssec:performance}

We evaluate the performance of our novel reduction algorithm and compare it to that of the baseline from the literature. For this, we measure the average time required to reduce eviction sets of different sizes to their minimal core.
We first focus on idealized conditions that closely match the assumptions of the theoretical analysis in Section~\ref{sec:algorithms}.

To put the performance of the reduction in context, we also evaluate the effort that is required for finding an initial eviction set to reduce. For this, we consider attackers with different capabilities to control the set index bits, based on huge pages ($\attackerbits=10$), 4\,KB pages ($\attackerbits=6$), and with no control over the set index bits ($\attackerbits=0$).

Together, our evaluation gives an account of how the performance gains of our novel reduction algorithm affect the overall effort of computing minimal eviction sets.

\paragraph*{Experimental Results}

The results of the evaluation of the reduction for a specific address on Skylake are given in Figure~\ref{fig:reduction_times}. We focus on cache set {\em zero} to mitigate the effect of the replacement policy, and we mitigate the influence of TLBs and prefetching as described in Section~\ref{subsec:design}.\footnote{In the {\em limit case} the stride of $64$B makes inferences by prefetching  prohibitive even with a randomized order, which is why we disable hardware prefetchers using \texttt{wrmsr -a 0x1a4 15}.}

Each data point is based on the average execution time of $10$ successful reductions. The sizes of the initial sets (x-axis) are chosen to depict the range where finding an initial eviction set does not require picking a too large number of candidate sets (depicted by the green bars). For a more systematic choice of the initial set size see the discussion below.

\begin{figure}[ht]
\includegraphics[width=\columnwidth]{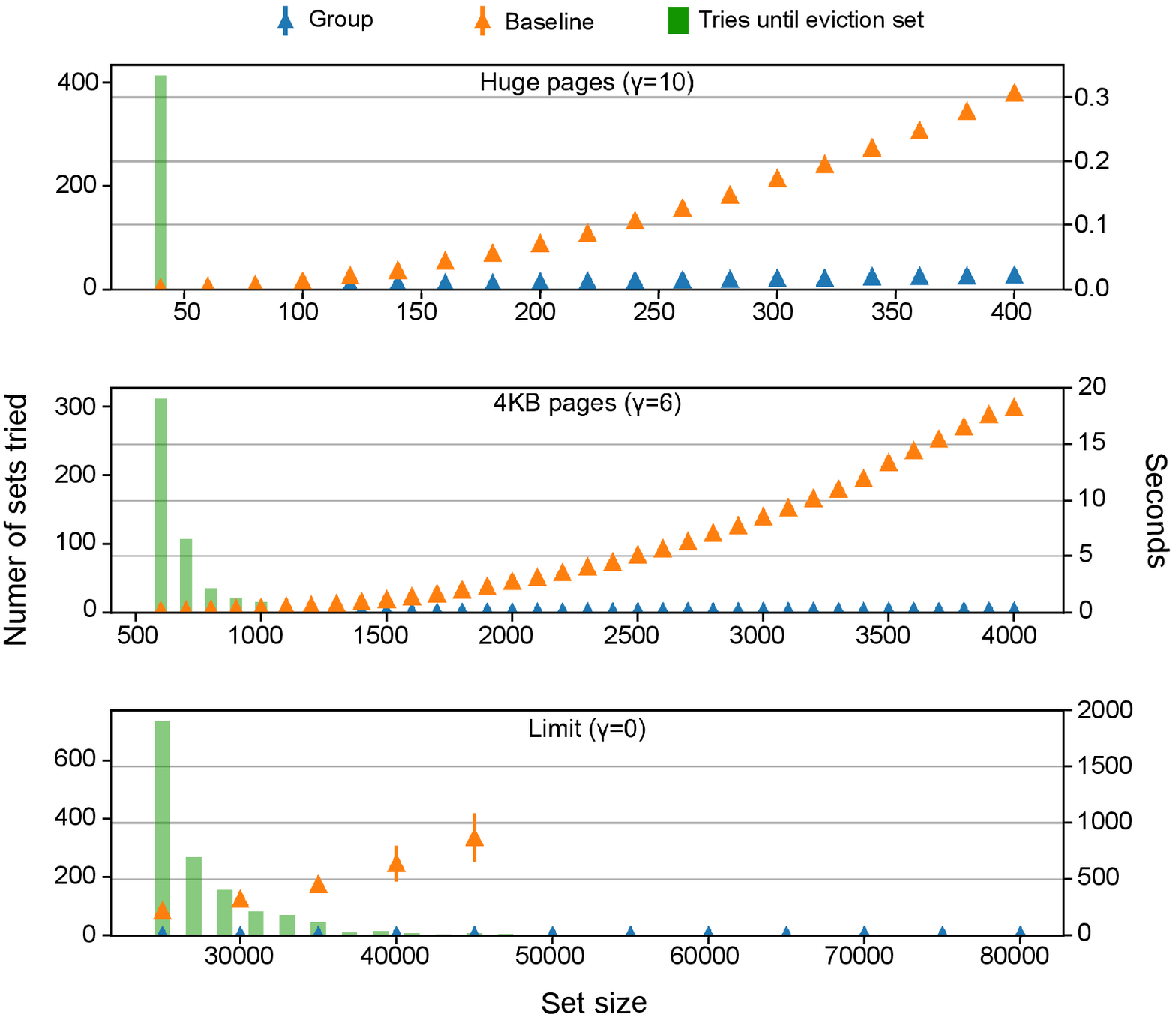}
\caption{The vertical green bars (left axis) depict the average number of times one needs to pick a set of addresses for finding an eviction set. Triangles (right axis) show time in seconds: {\em blue} depicts the average execution time of group test reductions; {\em orange} depicts the average execution time of baseline reductions. Different plots illustrate attackers with huge pages, $4$KB pages, and w/o any control over the set index bits, respectively.}
\label{fig:reduction_times}
\end{figure}

We highlight the following observations:
\begin{asparaitem}
\item The {\em slope} of the orange curve clearly illustrates the quadratic growth of execution time of the naive reduction, whereas the blue curve shows the linear growth of our novel algorithms. The {\em absolute values} account for constant factors such as the $50$ time measurements per test, and the overhead due to metrics collection.
\item For large set sizes, our novel reduction clearly outperforms the quadratic baseline. For example, for sets of size $3000$, we already observe a performance improvement of a factor of $10$, which shows a {\em clear practical advantage}. For small set sizes, this practical advantage seems less relevant. For such sizes, however, the number of repetitions required until find a real eviction set grows, as illustrated by the green bars. For the total cost of finding an eviction set, both effects need to be considered in combination.
\end{asparaitem}

\paragraph*{Optimal Choice of the Initial Set Size}

For evaluating the cost of first identifying and then reducing eviction sets, we rely on an expression for the overall number of memory accesses required for finding a minimal eviction set. This expression is the sum of the expected number $\frac{N}{p(N)}$ of memory accesses for finding an eviction set, see Section~\ref{subsec:probeviction}, and of the memory accesses for the respective reductions $N^2$ and $N$, see Propositions~\ref{prop:naive} and~\ref{prop:numacc}.
Based on this expression, we compute the optimal set sizes (from an attacker's perspective) for the linear and the quadratic reductions.  We use these sizes as an approximation of the optimal use of each reduction algorithm in the overall pipeline, and we evaluate their execution time on sets of this size.

Table~\ref{tab:optimal_sizes} shows the comparison of the linear and the quadratic reductions on sets of optimal size for three different attackers: with huge pages, with $4$KB pages, and in the limit.

\begin{table}[h]
\begin{tabular}{c | c | c | c | c | c |}
\multicolumn{2}{}{} & \multicolumn{2}{|c|}{Baseline} & \multicolumn{2}{|c|}{Group Testing} \\
\hline
Attacker & $P(C)$ & $N$ & Time & $N$ & Time \\
\hline
HP ($\attackerbits=10$) & $2^{-3}$ & $62$ & $0.005$s & $62$ & $0.004$s \\
$4$KB ($\attackerbits=6$) & $2^{-7}$ & $662$ & $0.179$s & $862$ & $0.023$s \\
Limit ($\attackerbits=0$) & $2^{-13}$ & $26650$ & $218.651$s & $53300$ & $1.814$s \\
\hline
\end{tabular}
\caption{$N$ shows the optimal set sizes for different attackers ($\attackerbits$ bits) on Skylake ($\assoc=12$) using $50$ time measurements per test. Time shows the average execution time of our implementations of Algorithm~\ref{alg:naivetesting} (baseline) and Algorithm~\ref{alg:grouptesting} (group testing) under ideal conditions.}
\label{tab:optimal_sizes}
\end{table}

We highlight the following observations:
\begin{asparaitem}
\item For huge pages, computing eviction sets is cheap, and the linear-time reduction does not lead to a practical advantage.
\item For small pages, the linear-time reduction improves the cost of computing eviction sets by a factor of more than 7. This is a significant advantage in practice, as it can make attacks more stealthy and robust against timing constraints.
\item For the limit case, the linear-time reduction improves over the quadratic baseline by more than two orders of magnitude.
\end{asparaitem}

\subsection{Performance in Practice}

In this section we give examples of the performance benefits of our reduction algorithms in real-world scenarios, i.e., in the presence of TLB noise and adaptive replacement policies.

We implement two heuristics to counter the resulting sub-optimal reduction rates (see~Section\ref{ssec:isolation}): {\em repeat-until-success}, where we pick a new set and start over after a failed reduction; and {\em backtracking}, where at each level of the computation tree we store the elements that are discarded, and, in case of error, go back to a parent node on which the test succeeded to continue the computation from there. For more details we refer to our open-source implementation.

For comparing the performance of the reduction algorithms in the context of these heuristics, we follow the literature and focus on initial set sizes that guarantee that the initial set is an eviction set with high probability. This is because a real-world attacker is likely to forgo the complications of repeatedly sampling and directly pick a large enough initial set.

The following examples provide average execution times (over $100$ samples) for different attackers on randomly selected target cache sets. Skylake ($\assoc=12$) using $10$ time measurements per test.
\begin{asparaitem}
\item For finding eviction sets with {\em huge pages}, previous work~\cite{MauriceHello2017} suggests an initial set size of $N=192$ which, according to our binomial model (see Section~\ref{subsec:probeviction}), yields a probability of sets to be evicting close to $1$.
For this size, the baseline reduction takes $0.014$ seconds, while the group-testing reduction takes $0.003$ seconds, i.e. our algorithm improves the baseline by a factor of $5$.
\item For finding minimal eviction sets with {\em $4$KB pages}, previous work~\cite{OrenSpy2015} suggests an initial set size of $N = 8192$, which amounts to the size of LLC times the number of slices. We choose an initial set size of $N = 3420$ for our comparison, which according to our model provides a probability of being an eviction set close to $1$.
For this $N$, the baseline reduction takes $5.060$ seconds, while the group-testing reduction takes $0.245$ seconds, i.e. our algorithm improves the baseline by a factor of $20$.
Finding {\em all} minimal eviction sets (for a fixed offset) within this buffer\footnote{We empirically observe that on Skylake, this size is sufficient to contain eviction sets for most of the $128$ different cache sets for a fixed offset.} requires more than $100$ seconds with the baseline algorithm. With group testing, the same process takes only $9.339$ seconds, i.e. it improves by a factor of $10$.
\end{asparaitem}

\subsection{Summary}
In summary, our experiments show that our algorithms improve the time required for computing minimal eviction sets by factors of 5-20 in practical scenarios. Moreover, they show that finding minimal eviction sets from virtual (or sandboxed) memory space is fast {\em even without any control} over the slice or set index bits, rendering countermeasures based on masking these bits futile.

\section{A Closer Look at the Effect of Modern Cache Replacement Policies}\label{sec:robustness}

There are several features of modern microarchitectures that are not captured in our model and that can affect the effectiveness of our algorithms, such as adaptive and randomized replacement policies, TLBs, prefetching, etc. The evaluation of Section~\ref{sec:evaluation} shows that the influence of prefetching can be partially mitigated by an adversary, and that the influence of TLBs is not a limiting factor in practice. The role of the cache replacement policy is less clear.

\begin{figure*}[ht]
\includegraphics[width=0.9\linewidth]{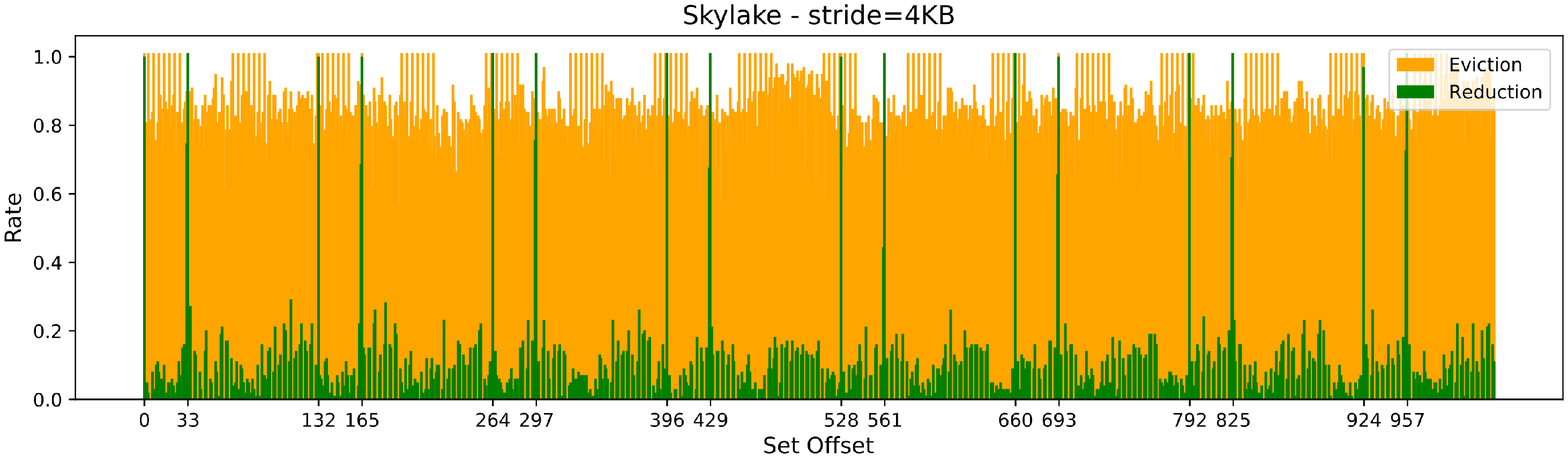}
\caption{Skylake's eviction and reduction rates per set index. With a stride of $4$KB and a total of $4000$ addresses (most of them non-congruent). The number of sets in-between two leaders is either 32 or 98. We rely on huge pages to precisely control the target's set index bits.}\label{fig:cachesets_skylake}
\end{figure*}

In this section, we take a closer look at the role of modern cache replacement policies in computing minimal eviction sets. As discussed in Section~\ref{sec:background}, post Sandy Bridge architectures boast replacement policies with features such as adaptivity or thrashing-resistance. With such features, accessing a set of $\assoc$ addresses that are congruent with $\setclass{x}$ is neither necessary nor sufficient for evicting $x$, which introduces a two-sided error (false positives and false negatives) in our tests for congruence. We first explain the key mechanisms that lead to this error, before we experimentally analyze its influence on Skylake and Haswell.

\subsection{Adaptive Replacement Policies}\label{ssec:replacement}
Adaptive cache replacement policies~\cite{Qureshi2006} dynamically select the replacement policy depending on which one is likely to be more effective on a specific load. For this, they rely on so-called {\em leader sets} that implement different policies. A counter keeps track of the misses incurred on the leaders and adapts the replacement policy of the {\em follower sets} depending on which leader is more effective at the moment. There are different  ways for selecting the leaders: a {\em static} policy in which the leader sets are fixed; and a {\em rand-runtime} policy that randomly selects different leaders every few millions instructions.

A previous study indicates that the replacement mechanism used in Ivy Bridge is indeed adaptive, with static leader sets~\cite{ReplacementPolicy2013}. To the best of our knowledge, there is no detailed study of replacement mechanisms on more recent generations of Intel processors such as Haswell, Broadwell, or Skylake, but there are mentions of high-frequency policy switches on Haswell and Broadwell CPUs as an obstacle for prime+probe attacks~\cite{OrenSpy2015}.

We perform different experiments to shed more light on the implementations of adaptivity in Skylake and Haswell, and on their influence on computing minimal eviction sets. To this end, we track eviction and reduction rates (see Section~\ref{sec:evaluation}) for each of the set indexes individually
\begin{enumerate}
\item on arbitrary eviction sets
\item on eviction sets in which all addresses are {\em partially} congruent.
\end{enumerate}
In the second case, the reduction uses only addresses belonging to a single cache set per slice. Assuming independence of cache sets across slice, a comparison with the first case allows us to identify the influence across cache sets.  For both experiments we rely on huge pages in order to precisely control the targeted cache set and reduce the effect of the TLB, see Section~\ref{ssec:isolation}.

\subsection{Evaluating the Effect of Adaptivity}
The results for reducing arbitrary eviction sets on Skylake are given in Figure~\ref{fig:cachesets_skylake}, the results for Haswell are given in Figure~\ref{fig:cachesets_haswell}. We focus on initial eviction sets of size $N = 4000$ (but observe similar results for other set sizes). We highlight the following findings:

\begin{figure}[ht]
\includegraphics[width=\linewidth]{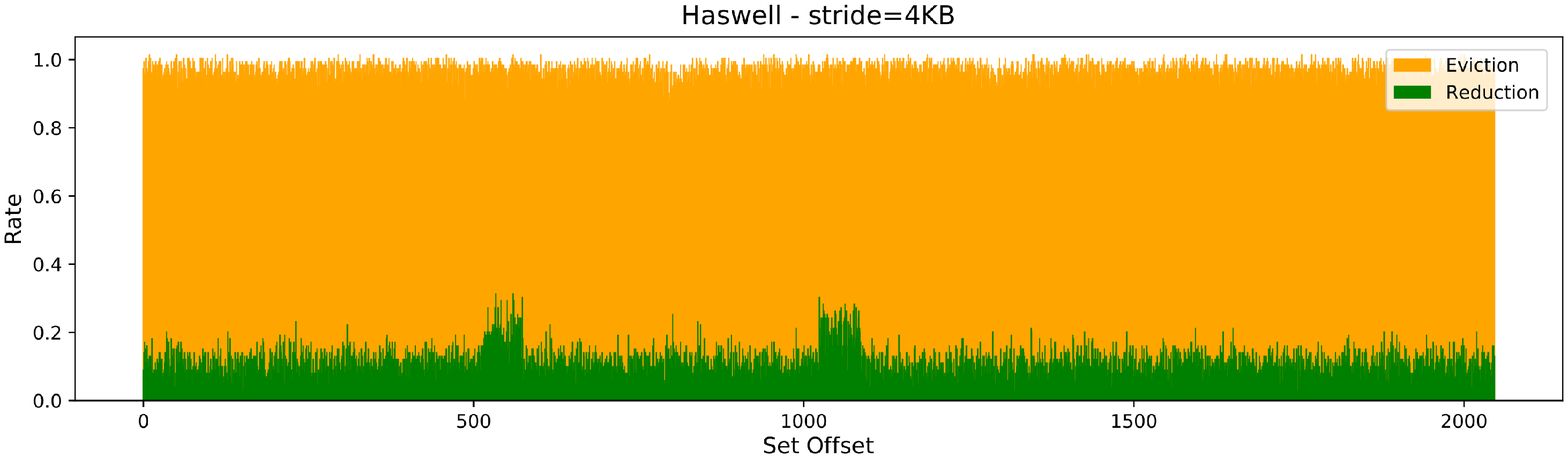}
\caption{Haswell's eviction and reduction rates per set index. With a stride of $4$KB and a total of $4000$ addresses (most of them non-congruent).}\label{fig:cachesets_haswell}
\end{figure}

\begin{figure}[ht]
\begin{subfigure}{\columnwidth}
    \includegraphics[width=\linewidth]{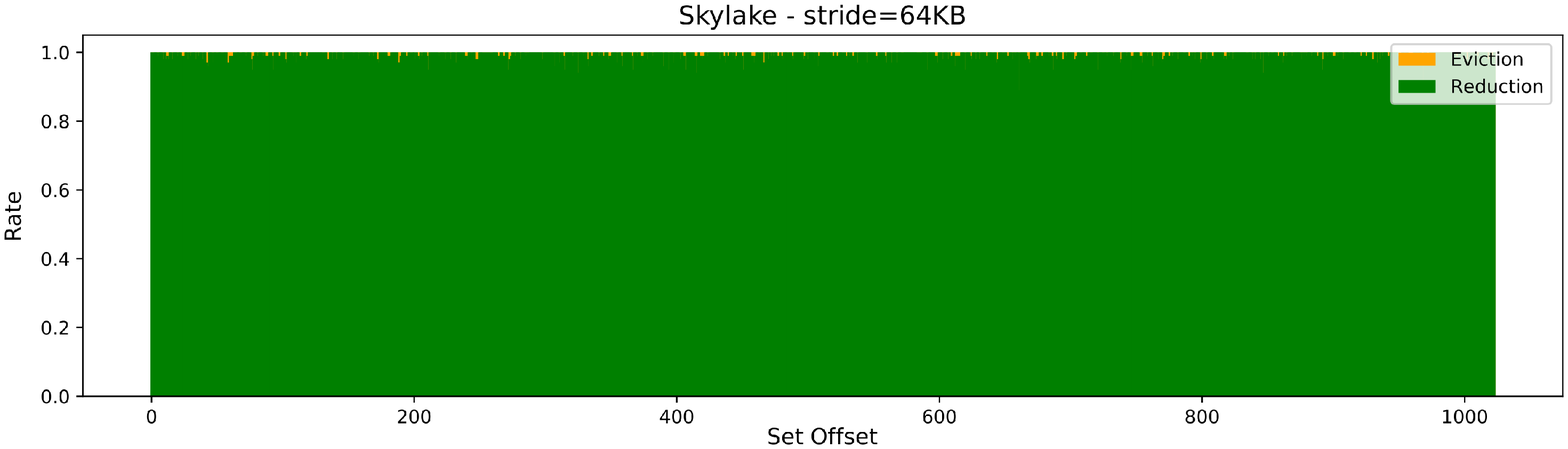}
    \caption{Skylake's eviction and reduction rates per set index, based on a stride of $64$KB (only partially congruent addresses).}
\end{subfigure}
\begin{subfigure}{\columnwidth}
    \includegraphics[width=\linewidth]{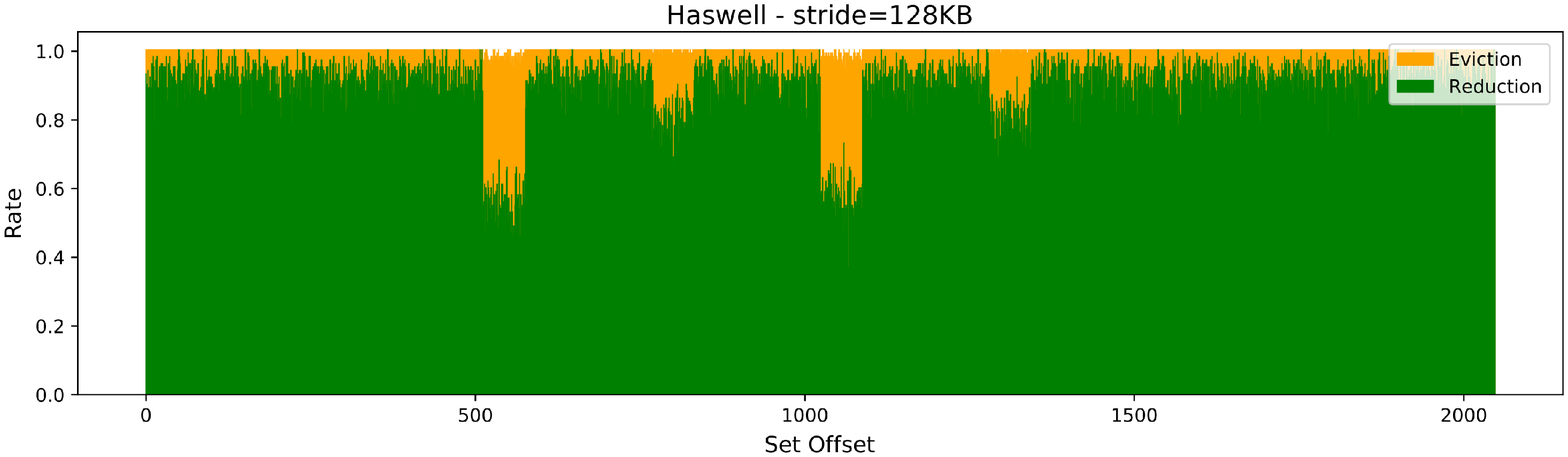}
    \caption{Haswell's eviction rate and reduction rate per set index, based on a stride of $128$KB (only partially congruent addresses).}
\end{subfigure}
\caption{Eviction rate and reduction rate per set index for initial sets of $4000$ partially congruent addresses. }
\label{fig:cachesets_work}
\end{figure}

\begin{figure*}[ht]
\includegraphics[width=0.9\linewidth]{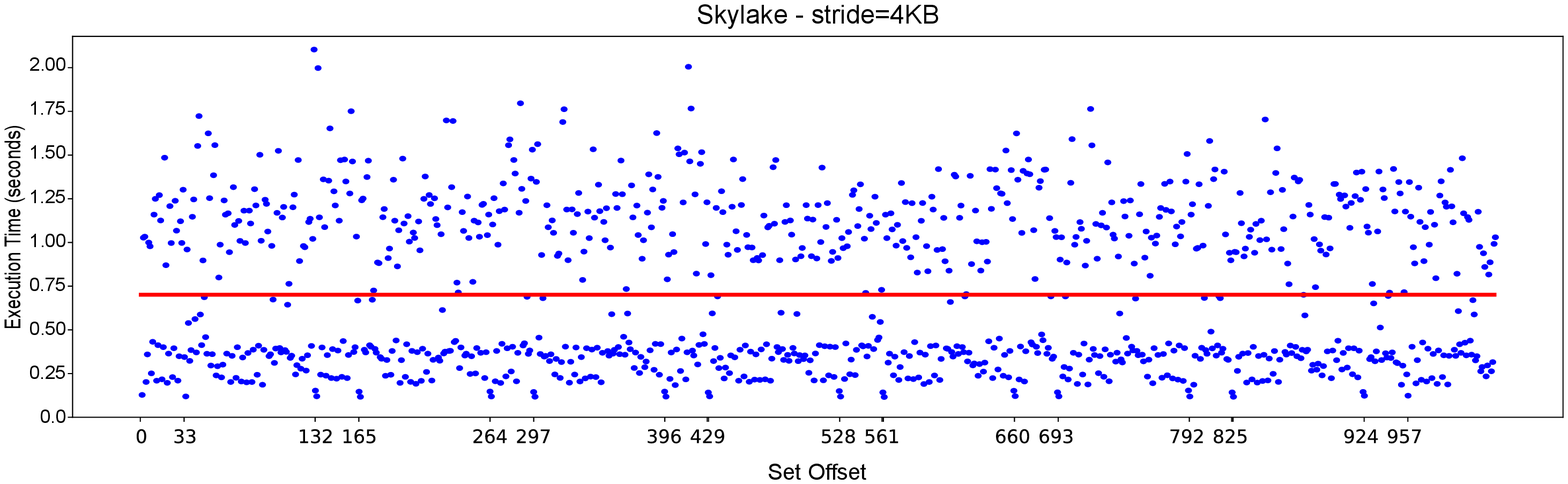}
\caption{Skylake's total execution time per set index using backtracking and repeat-until-success. Average time over 100 samples, all of them successful. Stride of $4$KB (simulate adversary) and initial set of $4000$ addresses (most of them non-congruent). The lowest execution times (below $0.12$s), correspond to sets with higher reduction rate. Horizontal line shows the overall average execution time.
}\label{fig:performance_skylake}
\end{figure*}

\begin{asparaitem}
\item Skylake seems to implement an adaptive replacement mechanism with static leader sets, much like Ivy Bridge. In particular, we identify a subset of 16 (out of 1024 per slice) sets where the reduction rate is consistently above $95$\% and where tests consistently evict the target address according to our model (i.e. without false positives and false negatives). On the follower sets the reduction rate quickly falls off despite a high eviction rate, indicating that the test can produce both false positives and false negatives.
\item In contrast to Skylake, on Haswell we are not able to identify any leader sets with consistently high reduction rates, which suggests a dynamic leader selection policy is in place.
\end{asparaitem}

The results of our reductions on partially congruent eviction sets on Haswell and Skylake are given in Figure~\ref{fig:cachesets_work}. They show that eviction and reduction rates are close to the predicted optimum. This improves over the reduction rate in Figure~\ref{fig:cachesets_skylake} and \ref{fig:cachesets_haswell}, and indicates a strong interference on the eviction test when accessing neighboring cache sets. In particular, we observe that the robustness of the reduction increases with the proportion of partially congruent addresses in the initial eviction set.

Finally, Figure~\ref{fig:performance_skylake} depicts the average execution time, including the overhead of the backtracking heuristic, of finding a minimal eviction set for each cache set index. A lower reduction rate implies a higher number of errors, and hence more backtracking steps and a longer execution time. This effect is visible when comparing with Figure~\ref{fig:cachesets_skylake}: cache sets with the highest reduction rates have the lowest execution times.\footnote{The plot also shows two different clusters of execution times, for which we currently lack a satisfying explanation.}

\subsection{Future Work}\label{ssec:future}
A more detailed analysis of the algorithmic implications of adaptive cache replacement policies is out of the scope of this paper. However, we briefly describe some ideas for future work:
\begin{asparaitem}
\item Controlling the policy. A better understanding of adaptivity mechanisms could be exploited to augment the number of followers that behave deterministically, and hence facilitate the reduction. For instance, once we find eviction sets for a leader set on Skylake, a parallel thread could produce hits on that set (or misses on another leader), ensuring that it keeps the lead.
\item Group testing. Work on {\em noisy} group testing~\cite{NoisyGroupTesting2017}, or threshold group testing with gap~\cite{Damaschke2006}, can provide useful tools for dealing with the uncertainty about the exact behavior of modern microarchitectures.
\end{asparaitem}

\section{Related Work}\label{sec:related}

Computing minimal, or at least small, eviction sets provides an
essential primitive for removing or placing arbitrary data in the
cache, which is essential for LLC cache attacks (Prime+Probe~\cite{LiuLLC2015},
Evict+Reload~\cite{GrussTemplate2015}, etc.), for DRAM fault attacks (such as
Rowhammer~\cite{KimDRAM2014,SeabornRH2015}, which break the
separation between security domains), for memory-deduplication attacks
(such as VUSION~\cite{OliverioVUSION2017}), as well as for the recent
Meltdown~\cite{meltdown18} and Spectre~\cite{spectre18} attacks
(which use the cache to leak data across boundaries and to increase
the number of speculatively executed instructions).

Gruss et al.~\cite{GrussRow2016} already identified {\em dynamic}
and {\em static} approaches for finding eviction sets. The dynamic
method uses timing measurements to identify colliding addresses
without any knowledge about the LLC slicing hash function or the physical
addresses; whereas the static method use the reverse engineered hash
and (partial) information about the physical addresses to compute
eviction sets.
%
%
In practice, most attacks~\cite{MauriceHello2017} rely in a hybrid
approach, producing a partially congruent set of addresses with static
methods, and pruning or reducing the results with a dynamic method
(mostly variants of Algorithm~\ref{alg:naivetesting}). We review some
of the most relevant approaches:

\paragraph*{Fully static, without slicing} In CPUs without slicing
(such as ARM) it is possible to find eviction sets directly using the
information from the pagemap interface. Lipp et al.~\cite{LippARM2016}
explore how to perform Prime+Probe, Evict+Reload, and other cross-core
cache attacks on ARM.
Fortunately, Google patched Android in March 2016~\footnote{Android patch:
  \url{https://source.android.com/security/bulletin/2016-03-01}}, and
now it requires {\em root} privileges to disclose physical addresses,
difficulting the task of finding eviction sets.

\paragraph*{Static/Dynamic with huge pages}
Liu et al.~\cite{LiuLLC2015} and Irazoqui et al.~\cite{Irazoqui2015},
in their seminal work on attacks against LLC,
rely on $2$MB huge pages to circumvent the problem of mapping virtual
addresses into cache sets.
They are the first to propose this method.

Gruss et al.~\cite{GrussRow2016} present the first rowhammer attack
from JavaScript. To achieve this, they build eviction sets thanks to
$2$MB huge pages (provided by some Linux distributions with {\em transparent
huge pages} support, see Appendix~\ref{apdx:hugepages}).
%

On the other hand, more sophisticated cache attacks from Intel's SGX~\cite{SchwarzSGX2017}
rely on the predictable physical allocation of large arrays within SGX enclaves,
and on the information extracted from another side-channel in DRAM row's buffers.

\paragraph*{Sandboxed environments without huge pages}
Oren et al.~\cite{OrenSpy2015} present an extension to Liu~et~al.'s work,
carrying out the first cache attack from JavaScript, where regular
$4$KB pages are used and pointers are not directly available.
It exploits the knowledge of browser's page aligned allocation for
large buffers to construct an initial set with identical page
offset bits. Then they leverage the clever technique described in
Section~\ref{ssec:evset4all} for further accelerating the process of
finding other eviction sets.

Dedup Est Machina~\cite{BosmanDedup2016} also implements a JavaScript
rowhammer attack, but this time targeting Microsoft Edge on Windows
10. Interestingly, they can not rely on large pages, since Microsoft
Edge does not explicitly request them. However, they discover that the
Windows kernel dispenses pages for large allocations from different
physical memory pools that frequently belong to the same cache sets.
Thereby, they are able to efficiently find eviction sets (not minimal) by
accessing several addresses that are $128$KB apart (and often land in
the same cache set).

Horn's~\cite{SpectreP0} breaks virtual machine isolation using a heuristic
to find small eviction sets by iterating over
Test~\ref{fig:any_test} several times, and discarding all elements
that are always {\em hot} (i.e. always produce cache hits). While this
heuristic performs extremely well in practice, its asymptotic cost is quadratic
on the set size.

Finally, a more recent work on cache attacks from portable
code~\cite{CachePortableGenkin2018} (PNaCl and WebAssembly) discusses
the problem of finding eviction sets on regular $4$KB pages and how to
partially deal with TLB thrashing.

In contrast to these approaches, our work is the first to consider
adversaries with less than $12$ bits of control over the physical
addresses, it formalizes the problem of finding eviction sets, and
provides new techniques that might \emph{enable purely dynamic approaches}.

 \paragraph*{Reverse engineering of slicing functions}
Modern CPUs~\footnote{According to Intel's Architecture Reference
  Manual~\cite{ia2018} (see {\em 2.4.5.3 Ring Interconnect and Last
    Level Cache}), Sandy Bridge is the first generation with slicing.}
with LLC slicing use proprietary hash functions
for distributing blocks into slices, which lead to attempts to reverse
engineer them.
These works are based on: 1) allocating and identifying sets of
colliding
addresses~\cite{SeabornSlice2015,Maurice2015RAID};
and 2) reconstructing the slice function using the hamming
distance~\cite{Hund2013}, or solving systems of
equations~\cite{IrazoquiSlices2015}, between these addresses.  Even
though we now know the slice hash function for several microarchitectures,
and Maurice et al.~\cite{MauriceHello2017} leverage it to speed up the finding of eviction
sets with huge pages, we believe that its use on real attacks is hindered
by constrained environments with scarce information about the physical addresses.

\paragraph*{Thrashing/scanning resistant replacement policies}
Modern replacement policies such as
insertion policies~\cite{QureshiJPSE07} or DRRIP~\cite{Jaleel2010},
are known to perform better than PLRU against workloads causing
scanning or thrashing. However, they also make eviction less reliable,
and fall outside our current models (see Section~\ref{sec:eviction}).
Howg~\cite{ReplacementPolicy2013} proposes a {\em dual pointer chasing}
to mitigate these effects; and Gruss~et~al.~\cite{GrussRow2016} generalize
the approach with {\em eviction strategies}, which are access patterns
over eviction sets that increase the chance of eviction under some unknown
modern policies. Both approaches are orthogonal to our in that they already
assume the possession of eviction sets.

\paragraph*{Set index randomization}
Concurrent work proposes some new designs for randomized caches~\cite{DBLPTrillaHAC18,Qureshi2018},
where cache sets are indexed with a keyed function that completely voids any
attacker control over the physical address bits. A key result of these proposals
is that they make cache-based attacks, and specially finding small eviction sets,
more difficult. Their security analysis, however, considers quadratic attackers;
it will be interesting to see how it is affected by our linear-time algorithm.

\section{Conclusion}\label{sec:conclusion}

Finding small eviction sets is a fundamental step in many microarchitectural
attacks. In this paper we perform the first study of finding eviction sets as
an algorithmic problem.

Our core theoretical contribution are novel algorithms that enable computing
eviction sets in linear time, improving over the quadratic state-of-the-art.
Our core practical contribution is a rigorous empirical evaluation in which
we identify and isolate factors that affect their reliability in practice,
such as adaptive replacement strategies and TLB thrashing.

Our results demonstrate that our algorithms enable finding small eviction sets
much faster than before, enabling attacks under scenarios that were previously
deemed impractical. They also exhibit conditions under which the algorithms
fail, providing a basis for research on principled countermeasures.

\vspace{1em}

\section*{Acknowledgments}
We thank Trent Jaeger, Pierre Ganty, and the anonymous reviewers for their
helpful comments. This work was supported by a grant from Intel Corporation,
Ram{\'o}n y Cajal grant RYC-2014-16766, Spanish projects
TIN2015-70713-R DEDETIS and TIN2015-67522-C3-1-R TRACES, and Madrid
regional project S2013/ICE-2731 N-GREENS.

\vspace{1em}

\bibliographystyle{ieeetr}
\bibliography{biblio}

\balance

\appendix
\subsection{Huge Pages}\label{apdx:hugepages}

Modern operating systems implement support for large buffers of virtual memory to be mapped into contiguous physical chunks of $2$MB (or $1$GB), instead that of regular $4$KB. These large chunks are called huge pages. On one hand, huge pages save page walks when traversing arrays of more than $4$KB, improving performance. On the other hand, they increase the risk of memory fragmentation, what might lead to wasting resources.

On Linux systems, huge pages can be demanded explicitly or implicitly:
\begin{itemize}
\item Explicit requests are done by passing special flags to the allocation routine (e.g. flag \verb!MAP_HUGETLB! to the \verb!mmap! function). In order to satisfy these requests, the OS pre-allocates a pool of physical huge pages of configurable size (which by default is $0$ in most systems).
\item Implicit requests are referred as {\em transparent huge pages}~\cite{THP}. THPs are implement with a kernel thread that, similarly to a garbage collector, periodically searches for enough contiguous $4$KB virtual pages that can be re-mapped into a free $2$MB chunk of contiguous physical memory (reducing PTs size). THP can be configured as: \verb!always!, meaning that all memory allocations can be re-mapped; \verb!never!, for disabling it; and \verb!madvise!, where the programmer needs to signal preference for being re-mapped via some flags (note that this is not a guarantee).
\end{itemize}

On other systems huge pages are implement differently, but the effect is generally the same. For instance, BSD's documentation refers to them as {\em super pages}, while Windows calls them {\em large pages}.

Interestingly, memory allocations in modern browsers are not backed by huge pages unless the system is configured with THP set to \verb!always!. Hence, relying on them for finding eviction sets is not feasible in most default systems.

\subsection{Proof of Proposition~\ref{prop:numacc}}\label{apdx:proofcost}

In the worst case, we access $\assoc+1$ different $\assoc$-subsets groups of size $\frac{n}{\assoc+1}$ each, and safely discard $\frac{n}{\assoc+1}$ elements that are not part of the minimal eviction set.
We first express recurrence~\eqref{eq:grouptesting} as a summation
$$ T(n) = \assoc n + \assoc n(\frac{\assoc}{\assoc+1}) + \assoc n(\frac{\assoc}{\assoc+1})^2 + ... + \assoc n(\frac{\assoc}{\assoc+1})^k $$
Our termination condition is $n(\frac{\assoc}{\assoc+1})^k < \assoc$,
meaning that we already have an minimal eviction set. By using the
logarithm we can set the exponent of the last iteration as
$k = \log_{\assoc/(\assoc+1)}{\assoc/n}$, which allows us to define
the function as a geometric progression
$$ T(n) = \assoc \sum_{i=1}^{i=k} n(\frac{\assoc}{\assoc+1})^{i-1} = \frac{\assoc n(1 - (\frac{\assoc}{\assoc+1})^{log_{\assoc/(\assoc+1)}{\assoc/n}})}{1 - \frac{\assoc}{\assoc+1}} $$
The logarithm in the exponent cancels out, and we obtain
$$ T(n) = \frac{\assoc n(1 - \frac{\assoc}{n})}{1-\frac{\assoc}{\assoc+1}} = \assoc(n-\assoc)(\assoc+1) = \assoc^2n + \assoc n - \assoc^3 - \assoc^2 $$

For simplicity we ignore the effect of the ceiling operator required in a real implementation, where $n$ is always an integer. It can be shown that this error is bounded by a small factor of our $k$, so we consider it negligible.

\subsection{Pidgeonhole Principle}\label{apdx:pidgeonhole}
\begin{proposition}
Let $\sizeof{S\cap P}\ge \assoc$ and let $T_1,\dots,T_{a+1}$ be a partition
of $S$. Then there is $i\in\{1,\dots,\assoc+1\}$ such that
$\sizeof{(S\setminus T_i)\cap P}\ge \assoc$
\end{proposition}
\begin{proof}
  When $\sizeof{S\cap P}=\assoc$, one of the $\assoc+1$ blocks of the
  partition must have an empty intersection with $P$, say $T_j$. Then
  $\sizeof{(S\setminus T_j)\cap P}=\sizeof{S\cap P}=\assoc$.

  When $\sizeof{S\cap P}> \assoc$, assume for contradiction that
\begin{equation}\label{eq:contradic}
\forall i\in\{1,\dots,\assoc+1\}\colon \sizeof{(S\setminus T_i)\cap
  P}<\assoc\ .
\end{equation}
We introduce the following notation: $p_i=\sizeof{T_i\cap P}$ and
$x=\sizeof{S\cap P}$. With this, \eqref{eq:contradic} can be
reformulated as
\begin{equation*}
\forall i\colon x-p_i <\assoc
\end{equation*}
Summing over all $i$ we obtain
\begin{equation*}
(\assoc+1)x - x= \assoc x <(\assoc+1) a\ ,
\end{equation*}
which contradicts the assumption that $ x > \assoc$.
\end{proof}

\subsection{Intel's TLBs}\label{apdx:tlbinfo}

Modern CPUs have very distinct TLBs implementations. In particular, modern Intel CPUs implement different buffers for data (dTLB) and instructions (iTLB), a second level TLBs (sTLB) with larger capacity, and different TLBs for each PT level.

Table~\ref{tab:tlbinfo} shows a summary of TLB parameters for Haswell and Skylake families:

\begin{table}[H]
\begin{tabular}{l | c | r }
& Haswell & Skylake \\
\hline
iTLB 4K & 128 entries; 4-way & 128 entries; 8-way \\
iTLB 2M/4M & 8 entries; full & 8 entries; full \\
\hline
dTLB 4K & 64 entries; full & 64 entries; 4-way \\
dTLB 2M/4M & 32 entries; 4-way & 32 entries; 4-way \\
dTLB 1G & 4 entries; 4-way & 4 entries; full \\
\hline
sTLB 4K/2M & 1024 entries; 8-way & 1536 entries; 4-way \\
sTLB 1G & - & 16 entries; 4-way \\
\hline
\end{tabular}
\caption{TLB implementation information for Haswell and Skylake microarchitectures. Extracted from the Intel's Architectures Optimization Manual~\cite{ia2018}.}
\label{tab:tlbinfo}
\end{table}

\clearpage

\end{document}